\documentclass[11pt, oneside]{article}    	% use "amsart" instead of "article" for AMSLaTeX format
\usepackage{geometry}                		% See geometry.pdf to learn the layout options. There are lots.
\usepackage{graphicx}				% Use pdf, png, jpg, or epsÂ§ with pdflatex; use eps in DVI mode
\usepackage{amssymb}
\usepackage{amsmath}
\usepackage{amsthm}
\usepackage{bbm}
\usepackage{booktabs}
\usepackage{dsfont}
\usepackage{color}
\usepackage{extarrows}
\usepackage{ulem}
\usepackage{enumerate}
\usepackage{caption}
\usepackage{subcaption}

\newtheorem{theorem}{Theorem}[section]
\newtheorem{lemma}{Lemma}[section]
\newtheorem{corollary}{Corollary}[theorem]

\newtheorem{algorithm}{Algorithm}

\theoremstyle{definition}
\newtheorem{definition}{Definition}[section]

\newcommand{\E}{\mathbb{E}}

\newcommand{\X}{\mathbf{X}}

\newcommand{\sd}[1]{\mbox{sd}{\left(#1\right)}}
\newcommand{\proj}[1]{\big|_{#1}}
\newcommand{\norm}[1]{\left\| #1 \right\|}
\renewcommand{\O}[1]{O\left(#1\right)}
\newcommand{\Ome}[1]{\Omega\left( {#1} \right)}
\newcommand{\Omep}[1]{\Omega_p\left( {#1} \right)}
\newcommand{\Op}[1]{O_p\left(#1\right)}

\newcommand{\changed}[1]{{#1}}
\newcommand{\chgx}[1]{{#1}}
\newcommand{\deleted}[1]{}

\newcommand{\blued}[1]{#1}
\newcommand{\bluedd}[1]{{#1}}

\newcommand{\gjh}[1]{{#1}}

\usepackage{natbib}

\begin{document}

\title{Boulevard: Regularized Stochastic Gradient Boosted Trees and Their Limiting Distribution}

\author{Yichen Zhou, Giles Hooker \\ Department of Statistics and Data Science, Cornell University}
\maketitle
\begin{abstract}
This paper examines a novel gradient boosting framework for regression. We regularize gradient boosted trees by introducing subsampling and employ a modified shrinkage algorithm so that at every boosting stage the estimate is given by an average of trees. The resulting algorithm, titled ``Boulevard'', is shown to converge as the number of trees grows. We also demonstrate a central limit theorem for this limit, allowing a characterization of uncertainty for predictions. A simulation study and real world examples provide support for both the predictive accuracy of the model and its limiting behavior.
\end{abstract}

\section{Introduction}

This paper presents a theoretical study of gradient boosted trees \citep[GBM or GBT:][]{friedman2001greedy}. Machine learning methods for prediction have generally been thought of as trading off both intelligibility and statistical uncertainty quantification in favor of accuracy. Recent results have started to provide a statistical understanding of methods based on ensembles of decision trees \citep{breiman1984classification}. In particular, the consistency of methods related to random forests \citep[RF:][]{breiman2001random} has been demonstrated in \citet{biau2012analysis,scornet2015} while
\citet{wager2014confidence, mentch2016quantifying, wager2017estimation} and \citet{wager2016generalized} prove central limit theorems for RF predictions. These have then been used for tests of variable importance and nonparametric interactions in \citet{doi:10.1080/10618600.2016.1256817}.

In this paper, we extend this analysis to GBT. Analyses of RF have relied on a subsampling structure to express the estimator in the form of a U-statistic from which central limit theorems can be derived. By contrast, GBT produces trees sequentially with the current tree depending on the values in those built previously, requiring a different analytical approach. While the algorithm proposed in \citet{friedman2001greedy} is intended to be generally applicable to any loss function, in this paper we focus specifically on nonparametric regression \citep{stone1977consistent, stone1982optimal}. Given a sample of $n$ observations $(x_1, y_1), \dots, (x_n, y_n) \in [0,1]^d \times \mathbb{R}$, assume they follow the relation
$$
X \sim \mu, \quad Y = f(X) + \epsilon
$$
which satisfies the following:
\begin{enumerate}[(M1)]
\item The density, $\mu$, is bounded from above and below, i.e. $\exists 0 < c_1 < c_2$ s.t. $c_1 \leq \mu \leq c_2$.
\item $f$ is bounded Lipschitsz, i.e. $|f(x)| \leq M_f < \infty$, and $\exists \alpha > 0$ s.t. $|f(x_1) - f(x_2)| \leq \alpha \norm{x_1 - x_2}, \forall x_1, x_2\in [0,1]^d$.
\item $\epsilon$ is sub-Gaussian error with $\E[\epsilon] =0$, $\E[\epsilon^2] = \sigma^2_{\epsilon}$, $\E[\epsilon^4] < \infty.$
\end{enumerate}

To fit this sample, GBT builds correlated trees sequentially to perform gradient descent in functional space \citep{friedman2000additive}. \blued{With $L^2$ loss, the process resembles the Robbins-Monro algorithm \citep{robbins1951stochastic}, applying trees to  iteratively fit residuals.} The procedure is given as follows.

\begin{itemize}
\item Start with the initial estimate $\hat{f}_0 = 0$.
\item For $b \geq 1,$ given the current functional estimate $\hat{f}_b$, calculate the negative gradient
$$
z_i \triangleq - \frac{\partial}{\partial u_i} \sum_{i=1}^n \frac{1}{2}\left(u_i - y_i\right)^2 \Big|_{u_i = \hat{f}_{b}(x_i)}= y_i - \hat{f}_b(x_i).
$$
\item Construct a tree regressor $t_b(\cdot)$ on $(x_1, z_1),\dots, (x_n, z_n)$.
\item Update by a small learning rate $\lambda > 0$,
$$
\hat{f}_{b+1} = \hat{f}_b + \lambda t_b.
$$
\end{itemize}

Gradient boosting is initially developed from attempts to understand Adaboost \citep{freund1999short} in \citet{friedman2000additive}. \citet{mallat1993matching} studied the Robbins-Monro algorithm and demonstrated convergence when the regressors are taken from a Hilbert space. \blued{Their result was extended by \citet{buhlmann2002consistency} to show the consistency of decision tree boosting under the settings of $L^2$ norm and early stopping.} From a broader point of view, discussions of consistency and convergence of general $L^2$ boosting framework can be found in \citet{buhlmann2003boosting}, \citet{zhang2005boosting} and \citet{buhlmann2007boosting}.

Besides the original implementation, there are a number of variations on the algorithm presented above. \citet{friedman2002stochastic} incorporated subsampling in each iteration and empirically showed significant improvement in predictive accuracy. \citet{rashmi2015dart} argued that GBT is sensitive towards the initial trees, requiring lots of subsequent trees to make an impact. Their solution incorporates the idea of dropout \citep{wager2013dropout, srivastava2014dropout} to train each new iteration with a subset of the existing ensemble to alleviate the imbalance. Further, \citet{rogozhnikov2017infiniteboost} suggested to downscale the learning rate and studied the convergence of the boosting path when the learning rate is small enough to guarantee contraction.

In this paper we intend to provide a universal framework combining \blued{the factors mentioned above} into GBT to study its asymptotic behavior. Our method is particularly inspired by the recent development of the RF inferential framework \citep{mentch2016quantifying, wager2017estimation, doi:10.1080/10618600.2016.1256817}, in which the averaging structure of random forests enables analyses based on U-statistics and projections leading to asymptotic normality. Similarly, among classic stochastic gradient methods, Ruppert-Polyak \citep{polyak1992acceleration, ruppert1988efficient} averaging implies asymptotic normality for model parameter estimators by averaging the gradient descent history. The boosting framework we present also produces models that exhibit this averaging structure which we can therefore leverage. We name this algorithm {\em Boulevard} boosting and summarize it in the following Algorithm \ref{alg:boulevard}.

\begin{algorithm}[Boulevard] \label{alg:boulevard} \hspace{1cm}
\begin{itemize}
\item Start with the initial estimate $\hat{f}_0 = 0$.
\item Given the current functional functional estimate $\hat{f}_b$, calculate the negative gradient
\begin{gather}
z_i \triangleq - \frac{\partial}{\partial u_i} \sum_{i=1}^n \frac{1}{2}\left(u_i - y_i\right)^2 \Big|_{u_i = \hat{f}_{b}(x_i)}= y_i - \hat{f}_b(x_i).
\label{fml:res}
\end{gather}
\item \blued{If needed,} generate a subsample $w \subseteq \{1,2,\dots, n\}$ \blued{(otherwise let $w = \{1,2,\dots, n\}$)}.
\item Construct a tree regressor $t_b(\cdot)$ on $\{ (x_i, z_i),i \in w\}$.
\item Update by learning rate $1 > \lambda > 0$, \blued{presumably closer to $1$}, that
\begin{gather}
\hat{f}_{b+1} = \frac{b-1}{b}\hat{f}_b + \frac{\lambda}{b} t_b =  \frac{\lambda}{b} \sum_{i=1}^{b} t_i.
\end{gather}
\item \blued{When boosting ends after the $b^*$-th iteration with $\hat{f}_{b^*}$, rescale it by $\frac{1+\lambda}{\lambda}$ to obtain the final prediction $\frac{1+\lambda}{\lambda} \hat{f}_{b^*}$.}
\end{itemize}
\end{algorithm}

\blued{As indicated by the last two steps, Boulevard and GBT differ in how they update the ensemble. The name {\em Boulevard} re-imagines the random forest stretched out along a row, with the averaging structure diminishing the influence of early trees as the algorithm proceeds.} The practical benefit of Boulevard \blued{update} is \blued{threefold}. First, shrinkage makes the ensemble less sensitive to any particular tree, particularly those at the start of the process. Second, subsampling reduces overfitting. \blued{Third, it is free of early stopping rules because there is always signal in the gradients.} As a result, the final form of the predictor sits between an ordinary GBT and a RF. \blued{Moreover, Boulevard provides us with theoretical guarantees regarding its limiting behavior. If we write $\hat{f}_{b, n}$ to be the functional estimate after $b$ iterations for a sample of size $n$, we can obtain the following results in their simplified forms.
\begin{enumerate}
\item Finite sample convergence. Conditioned on any finite sample of size $n$, should we keep boosting using Boulevard, the boosting path converges almost surely. In other words,  $\exists \hat{f}_{\infty, n}: \mathbb{R}^d \to \mathbb{R}$ relying on the sample such that for $a.s.$ $x$,
$$
\hat{f}_{b,n}(x) \to \hat{f}_{\infty, n}(x), b\to\infty.
$$
\item Consistency. Under several conditions regarding the construction of regression trees, this limit aligns with the truth after rescaling by $\frac{1+\lambda}{\lambda}$. For $a.s.$ $x$,
$$
\hat{f}_{\infty, n}(x) \xrightarrow{P} \frac{\lambda}{1+\lambda} f(x), n\to\infty,
$$
where the convergence in probability is with respect to sample variability.
\item Asymptotical normality. Under the same conditions for consistency, we can prove that the prediction is asymptotically normal. For $a.s.$ $x$,
$$
\frac{\hat{f}_{\infty, n}(x) - \frac{\lambda}{1+\lambda}f(x)}{\sd{\hat{f}_{\infty, n}(x)}} \xrightarrow{d} N(0,1), n \to \infty.
$$
\end{enumerate}
 }
\blued{We will demonstrate in detail these results along with their required conditions in later sections.} So far as we are aware, these represent the first results on a distributional limit for GBT and hence the potential for inference using this framework; we hope that they inspire further refinements. It is worth noticing that Bayesian Additive Regression Trees (BART) \citet{chipman2010} were also motivated by GBT and allow the development of Bayesian credible intervals. However, the training procedure for BART resembles backfitting a finite number of trees, resulting in a somewhat different model class.
%Nonetheless, we expect that some of the stochastic contraction mapping results developed below may be useful in demonstrating frequentist properties for the resulting BART estimators.

The remainder of the paper is organized as follows: \blued{In Section 2, we introduce {\em structure-value isolation}, a counterpart of honesty for boosted trees and {\em non-adaptivity} to ensure that the distribution of tree structures stabilizes. In Section 3, we show finite sample Boulevard convergence to a fixed point. In Section 4,} we further prove the limiting distribution. \blued{In Section 5, we focus on non-adaptivity and discuss how to achieve it either manually or spontaneously.} In Section 6, we present our empirical study.

%Without further specification, $||\cdot||$ is the operator norm when applied to matrices and the $L^2$ norm when applied to vectors and functions.

\section{\bluedd{Non-adaptivity}}
% Because of the mathematical difficulties of analyzing the greedy splitting rules of trees, most current analyses of RF and other ensemble methods have been based on variations of the tree building algorithm originally proposed in \citet{breiman2001random}.
The sequential process employed in constructing decision trees renders them challenging to describe mathematically and more tractable variants of the original implementation in \citet{breiman1984classification} are often employed. In particular, separating the structure of the tree from the observed responses -- which may still used to determine leaf values -- has been particularly useful.  Early analyses employed completely randomized trees in which the structure was determined independently of observed responses: see \citet{buhlmann2002analyzing,breiman2004consistency,biau2008consistency,biau2012analysis,arlot2014analysis} that, in particular, generate a connection to kernel methods observed in \citet{davies2014random} and \citet{scornet2016random}. \gjh{ See Section 3.1 of \citet{biau2016random} for an overview of these developments. }

More recent results in \citet{wager2017estimation} employ sample splitting so that the structure of a tree is determined by part of the data while the values in the leaves are decided by an independent set to create a property that is termed {\em honesty} -- independence between tree structure and leaf values. \citet{wager2017estimation} provide some evidence for the practical, as well as theoretical, utility of this construction. However, they also rely on a condition of {\em regularity}: requiring the sample probability of selecting each covariate to be bounded away from zero. This is important in controlling the size of leaves and therefore the bias of the resulting estimate.  This condition removes the variable selection properties of trees \citep[see][]{biau2012analysis} and we are unaware of practical implementations that guarantee it in practice.

In this paper, we rely on two restrictions on the tree building process: an extension of honesty that we call {\em structure-value isolation} to require the structure of trees independent of leaf values for the whole ensemble rather than individual trees,  and {\em non-adaptivity} that requires the distribution of tree structures (conditional on the sample) to be constant over boosting iterations. Non-adaptivity allows us to demonstrate the convergence of the Boulevard algorithm for a fixed sample and to show that its limit has the form of kernel ridge regression.  Structure-value isolation then allows us to provide distributional results by separating the form of the kernel from the response.

We recognize that these place restrictive conditions on the tree building process and note parallels between these and the early developments in the theory of random forests referenced above. Completely randomized trees satisfy both conditions but are not strictly necessary. Non-adaptivity can be replaced by an asymptotic version which  we argue may arise spontaneously or can be enforced in various ways. Structure-value isolation can be achieved, for example,  by global subsample splitting -- replacing the leaf values of trees by a held-out subset as Boulevard progresses, although note that our results do not cover uncertainty in the resulting kernel structure. Further details of these conditions are pursued below; we speculate in Section \ref{sec:alternative} that our distributional results may also be obtained via alternative representations, but developing these is beyond the scope of this paper.

\subsection{Structure-Value Isolation}
A decision tree \citep{breiman1984classification} predicts by iteratively segmenting the covariate space into disjoint subsets (i.e. leaves) within each of which a terminal leaf value is chosen for making predictions. Therefore we can separate the construction of a decision tree into two stages: deciding the tree structure, and deciding the values in the terminal nodes. Traditional greedy tree building algorithms use the same sample points for both of these two steps. One drawback of these algorithms is the difficulty of providing mathematical guarantees about isolating sample points with large observation errors, i.e. outliers, thereby de-stabilizing the resulting predicted values. We think of this behavior as ``chasing order statistics''. As a result, a plethora of conclusions on trees and tree ensembles rely on randomization, for example, using completely randomized trees or applying honesty as noted above.

In the context of Boulevard, the sequential dependence of trees requires us to employ a stronger condition. We define {\em structure-value isolation} as the requirement that leaf values are independent of tree structures across the entire ensemble, rather than on a tree-by-tree basis.  To make this precise, here we introduce our decision tree notations. A regression tree $t_n(\cdot)$ segments the covariate space $\Omega$ into disjoint leaves $\Omega = \bigsqcup_{j=1}^m A_j$ representing the tree structure. $\Omega = [0,1]^d$ and $\{A_j\}_{j=1}^m$ are hyper-rectangles. A traditional regression tree $t_n(\cdot)$ prediction can be explicitly expressed as
$$
t_n(x) = \sum_{i=1}^n s_{n,i}(x) y_i,
$$
where, given $x \in A_j$ for some $j$,
$$
s_{n,k}(x) = \frac{I(x_k \in A_j)}{\sum_{i=1}^n I(x_i \in A_j)},
$$
describing the influence of $x_j$ on predicting the value at $x$. Slight changes are required when a subsample is used instead of the full sample to calculate the leaf values. For given subsample $w \subset \{1,\dots, n\}$, we write
$$
t_n(x;w) = \sum_{i=1}^n s_{n,i}(x;w) y_i.
$$
In this case, for any $x \in A_j$,
$$
s_{n,k}(x)  = s_{n,k}(x; w) =\frac{I(x_k \in A_j)I(k \in w)}{\sum_{i=1}^n I(x_i \in A_j)I(i \in w)} = \frac{I(x_k \in A_j)I(k \in w)}{\sum_{x_i \in A_j} I(i \in w)}.
$$
We define $s_n(x) = (s_{n,1}(x), \dots, s_{n,n}(x))^T$ the (column) {\em structure vector} of x, and
$$
S_n =  \begin{bmatrix}
s_{n,1}(x_1) & \dots & s_{n,n}(x_1) \\
\vdots & \ddots & \\
s_{n,1}(x_n) & \dots & s_{n,n}(x_n)
\end{bmatrix} = \begin{bmatrix}s_n(x_1)^T \\ \vdots \\ s_n(x_n)^T\end{bmatrix}
$$
the {\em structure matrix} as the stacked structure vectors of the full sample. With this notation, structure-value isolation can also be viewed as the separation between the tree structure matrix and the response vector. Formally, we have

\begin{definition}
An ensemble of trees $$\hat{f}_{b, n}(x) = \sum_{j=1}^b t^j_n(x;w) = \sum_{j=1}^b \sum_{i=1}^n s_{n,i}^j(x;w) y_i$$ exhibits {\em structure-value isolation} if $s_{n,i}^j(x;w)$ is independent of the vector $Y=(y_1,\dots,y_n)^T$ for all $1\leq j \leq b$ and $x$, where we use the superscript $j$ to represent the $j$-th tree.
\end{definition}

For example, random forests utilizing completely randomized trees achieve structure-value isolation, while standard boosting and standard random forests do not.

\subsection{Non-adaptivity}

In addition to structure-value isolation, in order to demonstrate that Boulevard converges for a fixed sample, we also impose a condition which ensures that the distribution of tree structures stabilizes as Boulevard progresses. We describe this in two senses: a strong sense in which the distribution, conditional on the sample, is fixed for the whole sample, and a weak sense in which this occurs eventually.

\begin{definition} Denote $(Q_{n,b}, \mathcal{Q}_{n,b})$ the probability space of tree structures given sample $(x_1, y_1),\dots, (x_n, y_n)$ of size $n$ after $b$ trees have been built, where $Q_{n,b}$ consists of all possible tree structures and $\mathcal{Q}_{n,b}$ the probability measure on $Q_{n,b}$. A tree ensemble is {\em non-adaptive} if $(Q_{n,b}, \mathcal{Q}_{n,b}) = (Q_n, \mathcal{Q}_n)$ identical for all $b$. A tree ensemble is {\em eventually non-adaptive} if $(Q_{n,b}, \mathcal{Q}_{n,b})= (Q_n, \mathcal{Q}_n)$ identical for sufficiently large $b$.
\end{definition}

Following this definition, both standard random forests and random forests utilizing completely randomized trees are non-adaptive. In terms of boosting, since the target  responses change each iteration, a shortcut for achieving non-adaptivity is the aforementioned structure-value isolation. On the other hand, eventual non-adaptivity is a desirable condition should a GBT ensemble become stationary after enough iterations. The following algorithm provides a straightforward example to guarantee non-adaptivity.

\begin{algorithm}[Trees for Non-adaptive Boosting] \label{alg:honest} \hspace{1cm}

\begin{itemize}
\item Start with $(x_1, z_1),\dots, (x_n, z_n)$, where $z_1,\dots, z_n$ are current residuals.
\item Obtain the tree structure $q = \{A_j\}_{j=1}^m$ independently of $z_1,\dots, z_n$.
\item If needed, uniformly subsample an index set $w \subseteq \{1,\dots,n\}$ of size $\theta n$ (otherwise let $w = \{1,2,\dots, n\}$).
\item Decide the leaf values, hence $t_n(\cdot)$, merely with respect to $w$ as for $x \in A_j$,
$$
t_n(x) = \sum_{x_i \in A_j} \frac{I(i \in w)}{\sum_{x_l \in A_j} I(l \in w)} \cdot z_i,
$$
with $0/0$ defined to be $0$.
\end{itemize}
\end{algorithm}

Algorithm \ref{alg:honest} is not specific about how to decide the tree structures. For non-adaptivity it can be realized using completely randomized trees in which the gradients only influence the leaf values. Another strategy is to acquire another independent sample $(x'_1, y'_1),\dots, (x'_n, y'_n)$ and implement a random forest on this sample solely for extracting tree structures. We may also consider another alternative which uses a parallel adaptive boosting procedure on another independent sample to produce tree structures. While this option does not guarantee non-adaptivity, it results in trees with structure-value isolation and can be tailored to fulfill eventual non-adaptivity should we manage to reach an asymptotically stationary distribution of trees at the end of the parallel procedure.

In fact, eventual non-adaptivity allows the use of any tree building strategy for a finite number of trees. This assumption turns out to be sufficient for our theoretical development as long as we focus on the limit as the number of trees tends to infinite, rather than using a finite ensemble. Meanwhile it is compatible with any of the tree building algorithms used in practice for the beginning part where most signals get explored, which potentially accelerates the learning process. We will delay the discussion of fulfilling eventual non-adaptivity either spontaneously or manually to Section 5, and will focus on the Boulevard algorithm equipped with \blued{non-adaptivity} mechanism as {\em non-adaptive Boulevard} unless otherwise stated.

\subsection{Trees as Kernel Method}
\blued{Another motivation for isolating tree structures and leaf values comes from the fact that it relates decision trees to a kernel method. To clarify, recall the definition of the tree structure matrix and} denote by $(Q_n, \mathcal{Q}_n)$ the probability space of all possible tree structures given sample $(x_1, y_1),\dots, (x_n, y_n)$ of size $n$ and, when necessary, a subsampling framework $w$, where $q = \{A_i\}_{i=1}^{m_q} \in Q_n$ is the structure of a single possible tree. \blued{When we have structure-value isolation}, we can write the expected decision tree prediction on this sample as
\begin{align}
\label{fml:isolation}
\hat{Y} = \E_{q} \E_{w} [S_n] \cdot Y = \E_{q,w} [S_n] \cdot Y,
\end{align}
where $Y = (y_1,\dots, y_n)^T$, $\E_w$ the expectation with respect to subsample indices and $ \E_{q}$ the expectation with respect to the probability measure $\mathcal{Q}_n$.
\begin{theorem} \label{thm:treeaskernel} Denote by $\E_{q,w}$ the expectation over all possible tree structures and subsample index sets, then
\begin{enumerate}[(i)]
\item $\E_{q, w} [S_n] $ is symmetric, element-wise nonnegative.
\item $\E_{q, w} [S_n]$ is positive semidefinite.
\item $\norm{\E_{q, w}[S_n]}_1 \leq 1$, $\norm{\E_{q, w}[S_n]}_{\infty} \leq 1$, $\norm{\E_{q, w}[S_n]} \leq 1$. Here the last $\norm{\cdot}$ stands for the spectrum norm.
\end{enumerate}
\end{theorem}
\blued{This $\E_{q, w}[S_n]$ in (\ref{fml:isolation}) is similar to the random forest kernel \citep{scornet2015} defined by the corresponding tree structure space, subsampling strategy and tree structure randomization approach. In terms of a single tree, honesty contributes to its symmetry and positive semidefiniteness, while subsampling decides the concentration level of the kernel. While for an ensemble, non-adaptivity guarantees that  $\E_{q, w}[S_n]$ applies to all its tree components, making it possible to use a uniform mathematical expression to describe predictions from different trees.
\begin{align*}
\hat{Y_b} = \E_{q,w} [S_n] \cdot Y_b,
\end{align*}
for all $b$, where $Y_b$ is the target negative gradients and $\hat{Y_b}$ the fitted values after $b$-th iteration.}

\blued{It is also worth noticing that in part (iii) of Theorem \ref{thm:treeaskernel} we only guarantee inequality rather than setting the norm exactly to 1. This is caused by structure-value isolation: for instance, when applying subsample splitting, it is possible that certain leaf of the tree structure produced by the first subsample contains no points of the second subsample. For such cases we can predict 0 for expediency: \gjh{we demonstrate in Section \ref{sec:cos} that this choice has an asymptotically negligible effect below. }}

\section{Convergence}
\blued{As stated in \citet{zhang2005boosting}}, a first ``theoretical issue" of analyzing boosting method is the difficulty of attaining convergence. As a starting point we will show that Boulevard guarantees point-wise convergence under finite sample settings.

\subsection{Stochastic Contraction and Boulevard Convergence} To prove convergence of the Boulevard algorithm, we introduce the following definition, lemmas and theorem inspired by the unpublished manuscript by \citet{almudevarstochastic} regarding a special class of stochastic processes. We refer the readers to the original manuscript, but key points of the proof are briefly reproduced \blued{in Appendix \ref{sec:app:sc} and extended to cover a Kolmogorov inequality.}
\begin{theorem}[Stochastic Contraction]
\label{thm:stochasticcontraction}
Given $\mathbb{R}^d$-valued stochastic process $\{Z_t\}_{t \in \mathbb{N}}$, a sequence of $0 < \lambda_t \leq 1$, define
\begin{gather*}
\mathcal{F}_0 = \emptyset, \mathcal{F}_t = \sigma(Z_1,\dots, Z_t), \\
\epsilon_t =  Z_t - \E[Z_t | \mathcal{F}_{t-1}].
\end{gather*}
We call $Z_t$ a stochastic contraction if the following properties hold
\begin{enumerate}[(C1)]
\item Vanishing coefficients $$\sum_{t=1}^{\infty} (1-\lambda_t) = \infty, \mbox{ i.e. } \prod_{t=1}^{\infty}\lambda_t = 0.$$
\item Mean contraction $$||\E[Z_t|\mathcal{F}_{t-1}]|| \leq \lambda_t \norm{Z_{t-1}}, a.s..$$
\item Bounded deviation $$\sup \norm{\epsilon_t} \to 0, \quad \sum_{t=1}^{\infty}\E[\norm{\epsilon_t}^2] < \infty.$$
\end{enumerate}
In particular, a multidimensional stochastic contraction exhibits the following behavior
\begin{enumerate}[(i)]
\item Contraction $$Z_t \xlongrightarrow{a.s.} 0.$$
\item Kolmogorov inequality
\begin{align}
\label{fml:kolmax}
P\left( \sup_{t \geq T}\norm{Z_t} \leq \norm{Z_T} + \delta  \right) \geq 1- \frac{4\sqrt{d}\sum_{t=T+1}^{\infty}\E[\epsilon_t^2]}{\min\{\delta^2, \beta^2\}}
\end{align}
\blued{holds for all $T, \delta > 0$ s.t. }$\beta = \norm{Z_T} + \delta - \sqrt{d} \sup_{t > T} \norm{\epsilon_t} > 0$.
\end{enumerate}
\end{theorem}

\blued{To study Boulevard convergence, we start by informally postulating a convergent point $Y^*$ for any boosting algorithm. Ideally the boosting path should remain stationary at and after $Y^*$.

For standard boosting with $L^2$ loss, supposing we obtain $Y^*$ at the $(b-1)$-th iteration, the next update is $Y^* + \lambda_b S_b\cdot(Y-Y^*)$ with $S_b$ the tree structure matrix decided by the gradient vector. $Y^*$ being the convergent point implies that $\lambda_b S_b\cdot(Y-Y^*)$=0. Either we can implement a diminishing sequence of learning rates $\lambda_b \to 0$, or we must guarantee that we stop boosting with $S_b\cdot(Y-Y^*) \to 0$. The latter condition depends on both the properties of the gradient vector $Y-Y^*$ and the tree building algorithm deciding $S_b$, making it challenging to study either the existence of such point $Y^*$ or the convergence of the boosting path to $Y^*$.

For Boulevard, the update will be $\frac{b-1}{b} Y^* + \frac{\lambda}{b} S_b(Y-Y^*)$, implying that $Y^* = \lambda \E[S_b](Y-Y^*)$ when $S_b$ is random. For trees with non-adaptivity and structure-value isolation, we can solve this relation for $Y^*$ to prove its existence. Furthermore, Theorem \ref{thm:stochasticcontraction} can be used to study the convergence of the boosting path to such $Y^*$. We formalize our discussion as the following theorem.}
\begin{theorem}
\label{thm:boulevardconvergence} Given sample $(x_1, y_1),\dots,(x_n, y_n)$. If we construct gradient boosted trees with structure-value isolation and non-adaptivity using tree structure space $(Q_n, \mathcal{Q}_n)$, by choosing $M \gg \max\{M_f, y_1,\dots, y_n\}$ and defining $\Gamma_M(x) = sign(x)(|x| \land M)$ as a truncation function, let Boulevard iteration take form of
$$
\hat{f}_b(x) = \frac{b-1}{b}  \hat{f}_{b-1}(x) + \frac{\lambda}{b} s_b(x)(Y - \Gamma_M(\hat{Y}_{b-1})),
$$
where $Y = (y_1,\dots,y_n)^T$ the observed response vector, $\hat{Y}_{b} = (\hat{f}_b(x_1),\dots,\hat{f}_b(x_n))^T$ the predicted response vector by the first $b$ trees, $s_b$ the random tree structure vector. We have
$$
\hat{Y}_b \longrightarrow \left[\frac{1}{\lambda}I+\E[S_b]\right]^{-1}\E [S_b] Y,
$$
where $\E[\cdot] = \E_{q, w}[\cdot]$, $S_b$ the random tree structure matrix defined above.
\end{theorem}

This theorem guarantees the convergence of Boulevard paths under finite sample setting once we threshold by a large $M$.  As a corollary we can express the prediction at any point of interest $x$, \blued{which coincides with} a kernel ridge regression using the random forest kernel.
\begin{corollary} \label{cor:prediction} By defining $\hat{f} = \lim_{b \to \infty} \hat{f}_b$,
\begin{gather}
\hat{f}(x) = \E[s_n(x)] \left[\frac{1}{\lambda}I+\E[S_n]\right]^{-1} Y.
\label{fml:blvpred}
\end{gather}
\end{corollary}

Ridge regression tends to shrink the predictions towards 0 and so does (\ref{fml:blvpred}). Therefore \blued{Corollary \ref{cor:prediction} carries the message} that Boulevard may not converge to the whole signal. We will see in later proofs that Boulevard converges to $\frac{\lambda}{1+\lambda}f(x)$, in contrast to standard boosting where we expect the convergence to the full signal. In fact, Boulevard down-weighs the boosting history, thereby ensuring that each tree in the finite ensemble must predict some partial signal during training. It thus avoids being dominated by the first few trees then repeatedly fitting on noise. In practice, since we showed that the prediction from Boulevard is consistent with respect to $\frac{\lambda}{1+\lambda}f(x)$, we simply rescale it by $\frac{1+\lambda}{\lambda}$ to retrieve the whole signal.

\subsection{Beyond $L^2$ Loss} \label{sec:BeyondL2} Besides regression, other tasks may require alternative loss functions for boosting. For instance, the exponential loss $L(w,y) = \exp (-wy)$ in adaboost \citep{freund1995desicion}. Analogous to the proof for $L^2$ loss, we can write the counterparts for any general loss $L(u) = \sum_i L(u_i, y_i)$ whose non-adaptive Boulevard iteration takes the form of
$$
\hat{Y}_b = \frac{b-1}{b} \hat{Y}_{b-1} - \frac{\lambda}{b} S_b \nabla_w L(w) \Big|_{w=\hat{Y}_{b-1}}.
$$
Suppose the existence of the fix point $\hat{Y}^* = -\lambda \E[S_b] \nabla_w L(w) \Big|_{w=\hat{Y}^*}$, then
$$
\E[\hat{Y}_b - \hat{Y}^*|\mathcal{F}_{b-1}] = \frac{b-1}{b}(\hat{Y}_{b-1}-\hat{Y}^*) - \frac{\lambda}{b}\E[S_b]\left(\nabla_w L(w) \Big|_{w=\hat{Y}_{b-1}} - \nabla_w L(w) \Big|_{w=\hat{Y}^*}\right).
$$
If the gradient term is bounded and Lipschitz (which could be enforced by truncation), i.e. $$\norm{\nabla_w L(w) \Big|_{w=w_1} - \nabla_w L(w) \Big|_{w=w_2}} \leq M \norm{w_1-w_2},$$ we can similarly show such Boulevard iteration converges by choosing $\lambda \leq M^{-1}$. However, the closed form of $\hat{Y}^*$ can be intractable to obtain \changed{and potentially non-unique}. For example for AdaBoost, $\hat{Y}^*$ is the solution to $\hat{Y}^* = -\lambda \E [S_n] (\exp(-\hat{Y}^*_1 y_1),\dots,\exp(-\hat{Y}^*_n y_n))^T$.

\section{Limiting Distribution}
Inspired by recent results demonstrating the asymptotic normality of random forest predictions, in this section we prove the asymptotic normality of predictions from Boulevard. Before detailing these results, we need some prerequisite discussion on the rates used for decision tree construction in order to ensure asymptotic local behavior. In general, the variability of model predictions comes from \blued{three sources: the sample variability, response errors, and the \blued{systematic} viability of boosting caused by its subsampling and stopping rule. As shown in the last section, Boulevard convergence relieves the need for considering the stopping rule and the systematic viability. Therefore} the strategy for our proof is as follows: we first consider the \blued{boosting process conditioned on the sample triangular array so only the response errors contribute to the variability. We then establish the uniformity over almost all random samples to extend the limiting distribution to the unconditional cases,} showing that it is still the response errors that dominate the prediction variability.

\subsection{Building Deeper Trees}
Decision trees can be thought as k-nearest-neighbor \citep[k-NN:][]{altman1992introduction} models where $k$ is the leaf size and the distance metric is given by whether two points are in the same leaf. This adapts the metric to the local geometry of the response function. As the conclusions on k-NN predictions require growing-in-size and shrinking-in-radius neighborhoods \citep{gordon1984almost}, so are the counterparts of building deeper trees. Assuming non-adaptivity, the following assumptions are sufficient for our analysis.
\begin{enumerate}[(L1)]
\item Asymptotic locality. \blued{We introduce a sequence $\{d_n\}, d_n > 0, d_n \to 0$ by which the diameter of any leaf is bounded from above.} Writing $diam(A) = sup_{x, y \in A} |x-y|$, we require,
$$
\sup_{A \in q \in Q_n} diam(A) =  O(d_n).
$$ \label{cond:L1}
\item Minimal leaf size. \blued{We introduce another sequence $\{v_n\}, v_n > 0, v_n \to 0$ by which  the volume of any leaf is bounded from below.} Writing $V(\cdot)$ as the volume function by Lebesgue measure, we require that
$$
 \inf_{A \in q \in Q_n} V(A) = \Ome{v_n}.
$$ \label{cond:L2}
\end{enumerate}
 \blued{Recall the notation that $f = \Ome{g} \iff g = \O{f}$.} \gjh{These place requirements on all the leaves of a tree; we specify the rates we require for $d_n$ and $v_n$ below.}

These assumptions together bound the space occupied by any leaf of any possible tree from being either too extensive or too small. Along with the Lipschitz condition, this \blued{geometrically shrinking neighborhood not only indicates that each tree is asymptotically guaranteed to gather points with close response values together, but also helps to control the impact of one point on another across multiple consecutive iterations. To be specific, because boosting fits gradients instead of responses, a pair of distant points with divergent values might still contaminate each other through a bridge created by other points linked by being in the same leaf nodes across multiple iterations. Here (L\ref{cond:L1}) provides a theoretical upper bound for such influence. %We will later specify the rates we require in Boulevard proofs.

Notice that randomized honest tree building strategies can be compatible with (L1) if we force it, whereas the standard CART algorithm does not imply (L1). However, since the intuition of (L1) is to put points with similar responses together, we anticipate that CART, \gjh{tuned to a proper depth, will satisfy (L1), at least in practice}. If so, we can still exercise Boulevard boosting using trees built with greedy algorithms without compromising its theoretical guarantees in practice.}

\subsection{Conditioned on the Sample} \label{sec:cos} We first consider a sample triangular array, i.e. for each $n$, the sample $(x_{n,1},y_{n,1}),\dots,(x_{n,n}, y_{n,n})$ is given. The first subscript $n$ will be dropped when there is no ambiguity. We specify the rates for the size of leaf nodes as:
\begin{enumerate}[(R1)]
\item For some $\epsilon_1 > 0$,
$$
d_n = O\left(n^{-\frac{1}{d+2} - \epsilon_1}\right).
$$  \label{cond:R1}
\item For some $\epsilon_2 > \epsilon_1 > 0$,
$$
\inf_{A \in q \in Q_n} \sum_{i = 1}^n I(x_i \in A) = \Ome{n^{\frac{2}{d+2}-d \epsilon_2}}.
$$ \label{cond:R2}
\end{enumerate}
One compatible realization is
$$
d_n = O\left(n^{-\frac{1}{d+1}}\right), \quad \inf_{A \in q \in Q_n} \sum_{i = 1}^n I(x_i \in A) = \Ome{n^{\frac{1}{d+2}}}.
$$
For simplicity all our proofs are under this setting. However, any other rates satisfying these conditions are also sufficient.

Starting here we use the abbreviations
$$
k_n^T = \E[s_n(x)], \quad K_n = \E[S_n], \quad  r_n^T = k_n^T\left[\frac{1}{\lambda}I + K_n \right]^{-1}.
$$
\blued{to cover the components of (\ref{fml:blvpred}). The rates of $k_n$ and $r_n$ determine Boulevard asymptotics. }

Recall that the structure-value isolation paired with subsampling may lead to the consequence of predicting 0 in some leaves, so that we only guarantee $\norm{k_n}_1 \leq 1$. Working with the tree construction rates as above, the subsample rate $\theta$ determines how far away $\norm{k_n}_1$ is from 1. \blued{For Algorithm \ref{alg:honest},} without loss of generality, suppose we obtain a tree structure where each leaf contains no fewer than $n^{\frac{1}{d+2}}$ sample points before applying subsampling to decide the fitted values. If the subsample size is $\theta n = n^{\frac{d+1}{d+2}}\log n$, i.e. $\theta = n^{-\frac{1}{d+2}}\log n$, the chance of missing all sample points in one leaf is
\begin{align*}
p(n, \theta) & = \frac{\binom{n-n^{\frac{1}{d+2}}}{\theta n}}{\binom{n}{\theta n}}  = \frac{(n - \theta n)(n - \theta n-1)\cdots(n - \theta n-n^{\frac{1}{d+2}} + 1)}{n(n-1)\cdots(n-n^{\frac{1}{d+2}}+1)} \\
& \leq \left(\frac{n-\theta n}{n-n^{\frac{1}{d+2}}} \right)^{n^{\frac{1}{d+2}}} = \left(\frac{1-n^{-\frac{1}{d+2}}\log n}{1-n^{-\frac{d+1}{d+2}}} \right)^{n^{\frac{1}{d+2}}} \\
& \leq e\cdot \left(1-n^{-\frac{1}{d+2}}\log n \right)^{n^{\frac{1}{d+2}}} = \O{\frac{1}{n}}.
\end{align*}

Therefore, for any $x$, $1 - \norm{k_n}_1 = \O{\frac{1}{n}}$ if we use subsample size at least of $n^{\frac{d+1}{d+2}}\log n$. This requires the subsample to be relatively large, which is compatible with both $\theta$ being constant, or $\theta = (\log n)^{-1}$.

\subsection{Exponential Decay of Influence and Asymptotic Normality}
The prediction \chgx{that} Boulevard makes at a point is \chgx{a} linear combination of responses $y_1,\dots, y_n$ whose coefficients are given by $r_n$. \blued{Therefore the asymptotical normality of Boulevard predictions will rely on whether $r_n$ can comply with the Lindeberg-Feller condition. Since $r_n$ only differs from $k_n$ by a ridge regression style matrix multiplication, we can expect  $r_n$ hold similar properties as $k_n$ does. Further, by examining $r_n$ we can show that} distant points are less influential on the prediction, and this decay of influence  is exponential in our case.

Given any $n$-vector $v$ and an index set $D$, denote
$$
v \proj{D} = \begin{bmatrix}
v_1 \cdot I(1 \in D) \\
\vdots\\
v_n \cdot I(n \in D)
\end{bmatrix}.$$
\blued{This notation implies the decomposition that}
$v = v\proj{D} + v\proj{D^c}$.
\begin{lemma}
\label{lemma:expdecay}
Given sample $(x_1, y_1),\dots, (x_n, y_n)$, a point of interest $x$, set
$
l_n = \frac{\log n}{-\log \lambda} = c_1\log n,
$
and define index set $D_n = \{i : |x_i-x| \leq l_n\cdot d_n   \}$,
then
\blued{
\begin{enumerate}[(i)]
\item Globally,
$$
\left| \sum_{i=1}^n r_{n,i} - \frac{\lambda}{1+\lambda}\right| = \O{\frac{1}{n}}.
$$
\item If we only look nearby,
$$
\norm{r_n\proj{D^c_n}}_1 = O\left(\frac{1}{n}\right).
$$
\end{enumerate}
}
\end{lemma}

Lemma \ref{lemma:expdecay} indicates that Boulevard trees will asymptotically rely on a $\log n$ shrinking neighborhood around the point of interest. Given sample size $n$ and a point of interest $x$, \changed{we can therefore} define two neighborhoods of different radii $B_n = \left\{ i \big| |x_i - x| \leq  d_n \right\}$ and $D_n = \left\{ i \big| |x_i - x| \leq  l_n \cdot d_n \right\}$. $B_n$ contains all points that have direct influence on $x$ in a single tree, and $D_n$ contains the points that dominate the prediction at $x$. \blued{Their cardinalities} $|B_n|$ and $|D_n|$ follow Binomial distributions with parameters depending on the local covariate density.

To show the limiting distribution of for the sequence of predictions $\hat{f}_n(x)$, \blued{the key point is to verify that, by writing the prediction at any point $x$ as a linear combination of sample responses, all these coefficients are diminishing at a rate allowed by the Lindeberg-Feller condition. Therefore we take a look at the coefficient vectors $k_n$ and $r_n$.}
\begin{lemma}
\label{lemma:rateofk}
\blued{Given the triangular array with }sample $(x_{n,1}, y_{n,1}), \dots (x_{n,n}, y_{n,n})$ at size $n$, assume
\blued{
\begin{itemize}
\item The smaller neighborhood is growing fast enough: $\left| B_n \right|  = \Ome{ n \cdot d_n^d }$.
\item We have enough points in each leaf: $$\inf_{A \in q \in Q_n} \sum_{i = 1}^n I(x_{n,i} \in A) = \Ome{n^{\frac{1}{d+2}}}.$$
\end{itemize}
}
Then
$$
\Ome{n^{-\frac{1}{2}\frac{1}{d+1}}} = \norm{k_n}, \norm{r_n} = \O{ n^{-\frac{1}{2}\frac{1}{d+2}}}.
$$
\end{lemma}
\blued{Combining all our discussions above, the following theorem stating the asymptotic normality conditioned on a given sample sequence follows. Lemma \ref{lemma:rateofk} is used here to estimate the order of the variance term.}
\begin{theorem}[Conditional Asymptotic Normality for Boulevard Predictions]
\label{thm:fixed}
For any $x \in [0,1]^d$, \blued{suppose the conditions in Lemma \ref{lemma:rateofk} hold. } Write $f(X_n) = (f(x_1),\dots, f(x_n))^T$, \gjh{then under non-adaptivity and structure-value isolation,}
$$
\frac{\hat{f}_n(x) - r_n^Tf(X_n)}{\norm{r_n}} \xlongrightarrow{d} N(0,\sigma^2_{\epsilon}).
$$
\end{theorem}
%%\blued{The takeaway message implied by the theorem is that the variance of Boulevard prediction is simply the same as of a linear combination of response vectors, while the tree building strategy here guarantees the linear combination to be asymptotically normal.}
\subsection{Random Design}
\label{sec:rd}
To extend the scope of the conditional limiting distribution \blued{to the random design case where the covariates $x_1,\dots,x_n$ are randomly drawn from their distribution, we intend to claim all the assumptions we made leading to Theorem \ref{thm:fixed} also hold true for almost surely all random sample sequence. In order to do so, we first incorporate the sample randomness in a well-defined probability space regarding the triangular array constructed from a random sample sequence. This can be done by the Kolmogorov's extension theorem; see details in Appendix \ref{sec:ket}.

We also have to revise (R2) to account for sample randomness.} We increase the minimal leaf geometric volume by any small $\nu > 0$ s.t. $v_n$ follows $$v_n = \frac{n^{\frac{1}{d+2} + \nu}}{n} = n^{-\frac{d+1}{d+2} + \nu} < n^{-\frac{d}{d+1}} = \O{d_n^d}.$$ With this assumption, the following lemma demonstrates the asymptotic normality where the mean depends on the random sample.
\begin{lemma}
\label{lemma:random}
For given $x \in [0,1]^d$, suppose we have random sample $(x_1, y_1),\dots, (x_n, y_n)$ for each $n$. If we restrict the cardinality of tree space $Q_n$ using any small $\alpha > 0$ s.t.
$$
|Q_n| = \O{\frac{1}{n} \exp\left( \frac{1}{2} n^{\frac{1}{d+2}-\nu} - n^{\alpha}\right)},
$$
then
$$
\frac{\hat{f}_n(x) - r_n^Tf(X_n)}{\norm{r_n}} \xlongrightarrow{d} N(0,\sigma^2_{\epsilon}).
$$
\end{lemma}
The proof of Lemma \ref{lemma:random} allows us to substitute all $\O{\cdot}$ by $\Op{\cdot}$ in the analyses of random design. Further, \blued{we take advantage of the following theorem to} replace the empirical mean $r^T_nf(X_n)$ by its population version $\frac{\lambda}{1+\lambda}f(x)$. Combining all above we obtain the main theorem of this paper that the limiting distribution of the random design \blued{non-adaptive Boulevard predictions} is normal.
\begin{theorem}[Asymptotic Normality for Boulevard Predictions]
\label{thm:main}
For given $x \in [0,1]^d$, \gjh{under the conditions of Lemmas \ref{lemma:rateofk} and \ref{lemma:random} as well as non-adaptivity and structure-value isolation,}
$$
\frac{\hat{f}_n(x) - \frac{\lambda}{1+\lambda}f(x)}{\norm{r_n}} \xlongrightarrow{d} N(0,\sigma^2_{\epsilon}).
$$
\end{theorem}

\blued{Theorem \ref{thm:main} shows a deterministic mean term but a stochastic variance term.} From results on kernel ridge regression, we would expect that this stochastic variance converges in probability if the random forest kernel behaves as a generic kernel with a shrinking bandwidth. From a theoretical perspective, the optimal rate of $\norm{r_n}$ is bounded from below by $ n^{-\frac{1}{2}\frac{1}{d+1}}$, which corresponds to the optimal nonparametric regression rate using $\frac{1}{2}$-H\"{o}lder continuous functions as base learners \citep{stone1982optimal}. In practice, $\norm{r_n}$ relies on the specific method of growing the boosted trees, and therefore may vary from case to case.

Furthermore, this demonstrates that with carefully structured trees the prediction is consistent while the variance involves no signal but the error. It \changed{acts as} an undersmoothed local smoother whose bias term shrinks faster than the variance term.

\subsection{Subsampling and Tree Space Capacity}
We have a strict \changed{requirement that} the tree terminal \changed{node size grows at a rate} between $n^{\frac{1}{d+1}}$ and $n^{\frac{1}{d+2}}$ to guarantee undersmoothing. Any $\log$ term is allowed to be added to the existing polynomial result without changing the behavior. We notice that different subsample rates (i.e. $\theta = (\log n)^{-d}$ in \citet{wager2014confidence}, $\theta = n^{-1/2}$ in \citet{mentch2016quantifying}, although see relaxations in \citet{peng2019asymptotic}) have been applied for measuring uncertainty. In comparison, the Boulevard algorithm requires a relatively restricted rate between these. In addition, though Boulevard training implements subsampling at each iteration, this does not influence the asymptotic distribution. The impact of subsampling is on the \blued{possible variance} from the mean process therefore the convergence speed \blued{especially if we assume non-adaptivity}.

\blued{We have required the size of tree space to scale at a rate close to $\O{\frac{1}{n} \exp\left( \frac{1}{2} n^{\frac{1}{d+2}}\right)}$.} \blued{In comparison}, \citet{wager2015adaptive} have shown that, in fixed dimension, any tree can be well approximated by a collection of $O(\exp(\log n)^2)$ hyper rectangles.  \blued{One feasible way to compare these two rates is to consider a tree space satisfying Boulevard rates in which each tree has several leaves that constitute one hyper rectangle in the collection for approximation. Therefore a tree space of cardinality $O(\exp(\log n)^2)$ is enough for approximating any tree ensemble. In this sense, the capacity of our designated tree space is large enough from a practical perspective.}

\subsection{Alternative Analyses} \label{sec:alternative}

The distributional results above are based on an analysis of a kernel representation of the Boulevard process in which structure-value isolation plays a key role in both consistency and asymptotic normality results.  By contrast, \citet{mentch2016quantifying} and \citet{wager2017estimation} employed results derived from U-statistics to demonstrate asymptotic normality when random forest trees are obtained using subsamples.  This analysis uses the subsample structure to bypass the kernel representation and allows the trees to be built adaptively.  In a similar manner,  when Boulevard employs subsampling we may represent its fixed point as a $U$ statistic. However, as in our discussion in Section \ref{sec:BeyondL2}, $U$-statistic kernel is then defined implicitly with respect to a non-smooth tree-building process and further analysis is beyond the scope of this paper.

\citet{wager2017estimation} also uses regularity to constrain leaf sizes in order to demonstrate consistency outside of the $U$-statistic framework. We expect that a similar weakening of these conditions is possible in Boulevard.

\section{Eventual Non-adaptivity}
All the  results discussed above have assumed the non-adaptivity of the boosting procedure of Boulevard. In standard boosting however, it is conventional and \blued{important} to decide tree structures on the current gradients in order to better exploit the gap between the prediction and the signal. Such procedures are known for their tendency to overfit, behavior which can be relieved by subsampling. However, when seeking to extend our results to this case we lose the easy identifiability of a Boulevard convergence point as the tree structure distribution changes at each iteration. We therefore need more assumptions and further theoretical development to extend our convergence results to a more practical Boulevard algorithm that allows the current gradient to determine tree structure.

A first approach to this is to relax non-adaptivity to eventual non-adaptivity. We postulate a convergent sequence of predictions, indicating that the underlying tree spaces will be stabilized after boosting for a sufficiently long time. \blued{If this is the case, eventual non-adaptivity is spontaneous.} Here we introduce the notation $\E[S_n(Y, \hat{Y})]$ where $Y = (y_1,\dots, y_n)^T$ and $\hat{Y} = (\hat{f}(x_1),\dots,\hat{f}(x_n))^T$ indicating the expected tree structure given the gradient of the loss between observed responses and current predictions. In regression this is $Y - \hat{Y}$, and we will take this form into the following discussion instead of a generic gradient expression.

\subsection{Local Homogeneity and Contraction Regions}
We start with trees whose splits are based on the optimal \bluedd{L2 gain} \citep{breiman1984classification}. For $(x_1, z_1), \dots, (x_n, z_n)$, the chosen split minimizes the impurity in the form of
\begin{align}
\label{fml:gini}
\inf_{L, R} \quad \sum_{i \in L} (z_i - \bar{z}_{L})^2 + \sum_{i \in R} (z_i - \bar{z}_{R})^2,
\end{align}
where $L \subset \{1,\dots, n\}, R = L^C$. Once the optimal split is unique, i.e. the optimum has a positive margin over the rest, we could allow a small change of \changed{all $y$'s values} without \deleted{altering}\changed{changing} the split decision. \deleted{So is true universally if we introduce subsampling to a fixed sample and consider all possible subsample split decisions which are of finite amount.}\changed{This also holds true if the split is decided by a subsample.} In terms of adaptive boosting, this observation demonstrates {\em local homogeneity} that, except on a set $\Omega_0 \subset \mathbb{R}^n$ with Lebesgue measure 0 where $(z_1,\dots, z_n)^T = Y - \hat{Y} \in \Omega_0$ has multiple optima for (\ref{fml:gini}), \changed{we can segment $\mathbb{R}^n$, the space of possible vectors $Y - \hat{Y}$}, into subsets $\bigsqcup_{i=1}^{\alpha} C_i = \mathbb{R}^n \backslash \Omega_0$ s.t. $\E[S_n(Y, \hat{Y})] = \E[S_n(Y, \hat{Y}')]$ for $Y - \hat{Y}, Y-\hat{Y}' \in C_i$ the same subset.

Notice that Gini gain is insensitive to multiplying $(y_1,\dots, y_n)$ by a nonzero factor. Therefore all $C_i$'s are open double cones in $\mathbb{R}^n$.

\begin{definition}[Contraction Region] Given the sample $(x_1,y_1),\dots (x_n, y_n)$. Write $Y = (y_1,\dots,y_n)$ and current prediction $\hat{Y} = (\hat{y}_1,\dots,\hat{y}_n)$. Following the above segmentation $\bigsqcup_{i=1}^{\alpha} C_i = \mathbb{R}^n \backslash \Omega_0$. We call $C_i$ a contraction region if $Y^* \in C_i$ for the following $Y^*$ $$
Y^* = \lambda \E[S_n(Y, \hat{Y})](Y - Y^*), \mbox{ i.e. } Y^* = \left[ \frac{1}{\lambda}I + \E[S_n(Y, \hat{Y})]\right]^{-1}\E[S_n(Y, \hat{Y})]Y,$$
for any $Y-\hat{Y} \in C_i$, where $\E[S_n(Y, \hat{Y})]$ is the unique structural matrix in this region.
\end{definition}

The intuition behind this definition is that, as long as a Boulevard process \changed{stays} inside a contraction region, the subsequent tree structures will be conditionally independent of the predicted values. Therefore the process \chgx{becomes} non-adaptive, \deleted{and as per our previous discussion collapses}\changed{collapsing} to $Y^*$. To achieve this eventual non-adaptivity, we would like to know when a Boulevard path is \changed{permanently contained} in a contraction region.

We should point out here that we have not shown the existence and the uniqueness of such contraction regions. Such \changed{an analysis would rely} on the method of choosing splits, the training sample and the choice of $\lambda$.

\subsection{Escaping the Contraction Region}

In this section we explore possible approaches to restrict a Boulevard process inside a contraction region. Assuming the existence of contraction regions, we recall Theorem \ref{thm:stochasticcontraction} which indicates that the Boulevard process has positive probability of not moving far from the fixed point.

\begin{theorem}
\label{thm:escape}
Denote $B(x, r)$ the open ball of radius $r$ centered at $x$ in $\mathbb{R}^{n}$. Suppose $C \subset \mathbb{R}$ a contraction region, $Y^* \in C$ the contraction point and $B(Y, 2r) \subset C$ for some $r>0$. Write $\hat{Y}_b$ the Boulevard process. For sufficiently large $t$,
$$
P\left( \hat{Y}_b \in C, \forall b \geq t \big| \hat{Y}_t \in B(Y^*, r)\right) \longrightarrow 1, t \to \infty.
$$
\end{theorem}

Theorem \ref{thm:escape} \blued{states that the longer we boost, the more likely the boosting path is trapped in a contraction region forever, resulting in spontaneous eventual non-adaptivity and justifying all our theory assuming non-adaptivity.} However, it guarantees neither the existence or the uniqueness of the contraction region.

\subsection{Forcing Eventual Non-Adaptivity} \blued{If we alternatively consider manually forcing eventual non-adaptivity}, a possible \textit{ad hoc} solution to the existence is to apply a {\em tail snapshot} which \blued{uses the tree space of certain iteration $b^*$ for the rest of the boosting run. This $b^*$ can be either pre-specified or decided on the fly when the Boulevard path becomes stationary. Ideally, as long as $b^*$ is in a contraction region, its forced non-adaptivity is indistinguishable from  spontaneous eventual non-adaptivity. An example of Boulevard regression implementing the tail snapshot strategy is detailed in Algorithm \ref{alg:tailsnapshot}, in which we decide $b^*$ by examining the training loss reduction.}

\begin{algorithm}[Tail Snapshot Boulevard] \label{alg:tailsnapshot} \hspace{1cm}
\begin{itemize}
\item Start with $\hat{f}_0 = 0$ and \blued{$b^*$ to be specified.}
\item For $b = 0,\ldots$, given $\hat{f}_b$, calculate the gradient
$$
z_i \triangleq - \frac{\partial}{\partial u_i} \sum_{i=1}^n \frac{1}{2}\left(u_i - y_i\right)^2 \Big|_{u_i = \Gamma_M(\hat{f}_{b}(x_i))}= y_i - \Gamma_M(\hat{f}_b(x_i));
$$
\item If $b^*$ is not specified, update by $1>\lambda > 0$ and the tree structure space $Q_b$ decided by current gradients,
$$
\hat{f}_{b+1}(x) = \frac{b}{b+1}  \hat{f}_{b}(x) + \frac{\lambda}{b+1} s_b(x; Q_b)(z_1,\dots,z_n)^T,
$$
where $s_b(x; Q)$ denotes the random tree structure vector based on tree space $Q$. If $b^*$ is specified, update by $Q_{b^*}$ instead, i.e.
$$
\hat{f}_{b+1}(x) = \frac{b}{b+1}  \hat{f}_{b}(x) + \frac{\lambda}{b+1} s_b(x; Q_b^*)(z_1,\dots,z_n)^T.
$$
\item When $b^*$ is not specified, check the empirical training loss as a measure of the distance to the fixed point.
$$
L_{b+1} = \frac{1}{2n}\sum_{i=1}^n \left(\frac{\lambda}{1+\lambda}y_i - \hat{f}_{b+1}(x_i)\right)^2.
$$
If $L_{b+1} < L^*$ a preset threshold, we claim Boulevard is close enough to a fixed point and specify $b^*$ to be current $b+1$. \blued{This step is then skipped for the rest of the boosting run.}
\end{itemize}
\end{algorithm}

\blued{Another plausible solution is to collectively build the tree space up based on the following intuition. For the first trees we prefer to construct them using the gradients to capture the signal. However when the boosting process keeps going on, we may consider reusing the tree structures or the tree spaces of previous iterations as to avoid overfitting the data or fitting the errors. Therefore the tree space $Q_{n,b}$ for a large $b$ should be the tree structures created by current gradients, plus all tree spaces $Q_{n,i}$ for $i<b$. This accumulation scheme changes as we boost, however asymptotically it leads to a common tree space and thus eventual non-adaptivity.}

\section{Empirical Study}

We have conducted an empirical study to demonstrate the performance of Boulevard. Despite the fact that \deleted{design purpose of Boulevard falls on the statistical inference side}\chgx{our purpose in developing Boulevard lies in statistical inference}, we require its accuracy to be on par with other predominant tree ensembles, which is assessed on both simulated and real world data. In addition, we inspect the empirical limiting behavior of non-adaptive Boulevard to show its \changed{agreement} with our theory. We summarize the result of the empirical study in this section, while \changed{additional} details can be found in Appendix \ref{sec:emp}.

\subsection{Predictive Accuracy}
We first compare Boulevard \changed{predictive accuracy} with the following tree ensembles: \deleted{ on predictive accuracy. Abbreviate the tree ensembles as follows: }Random Forest (RF), gradient boosted trees without subsampling (GBT), stochastic gradient boosted trees (SGBT), non-adaptive Boulevard achieved by completely randomized trees (rBLV), adaptive Boulevard whose tree structures are influenced by gradient values (BLV). All the tree ensembles build same depth of trees \changed{(see Appendix \ref{sec:appemp} for details)} throughout the experiment.

Results on simulated data are shown in Figure \ref{fig:sim_syn}. We choose sample size of 5000 and use the following two settings as underlying response functions: (1) $y = x_1 + 3x_2 + x_3x_4$ (top), and (2) $y = x_1 + 3x_2 + (1-x_3)^2 + x_4x_5 + (1-x_6)^6 + x_7$ (bottom). Error terms are Unif[-1,1] (left) and equal point mass on $\{-1, 1\}$ (right). Training errors are evaluated on the training set with noisy responses, while test errors are evaluated by error from the underlying signal on a separate test set. \blued{For each setup we have two plots covering the behavior of the first 50 trees and the following 250 trees respectively.} BLV and rBLV are comparable with RF, while all the three equal-weight tree ensembles are slightly inferior to GBM and SGBM.

\begin{figure}[h] %  figure placement: here, top, bottom, or page
   \centering
   \includegraphics[width=0.70\textwidth]{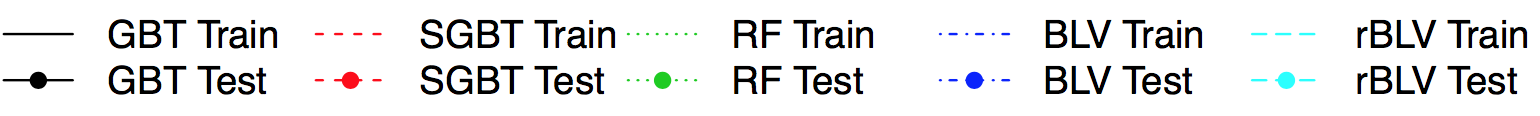}
   \begin{subfigure}[h]{0.49\textwidth}
   	\centering
        \includegraphics[width=0.375\textwidth]{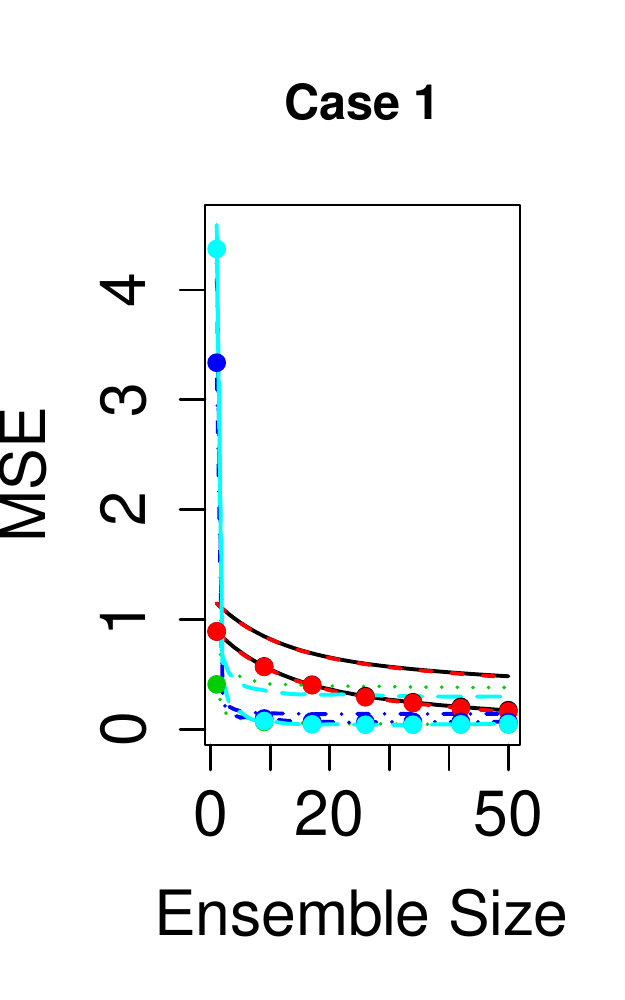}
        \includegraphics[width=0.60\textwidth]{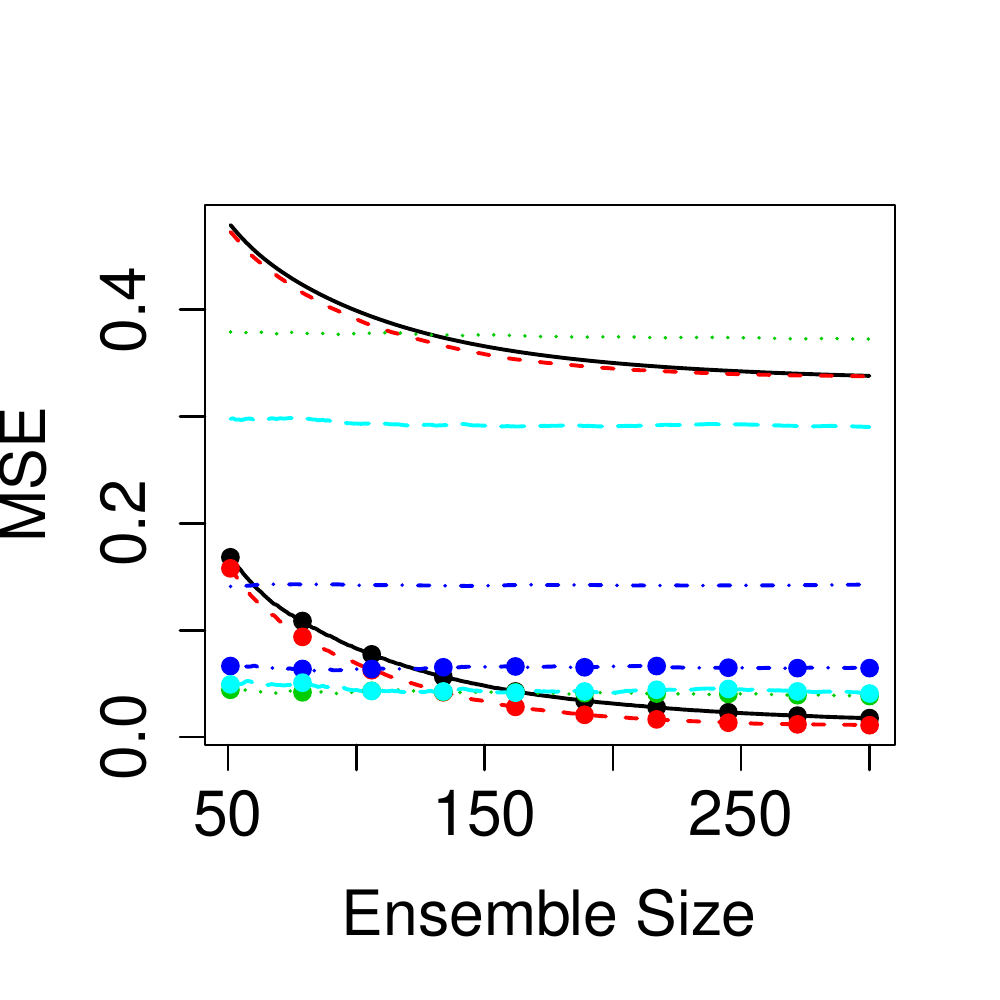}
        \caption{Function (1), Unif[-1,1] error.}
   \end{subfigure}
   \begin{subfigure}[h]{0.49\textwidth}
   	\centering
        \includegraphics[width=0.375\textwidth]{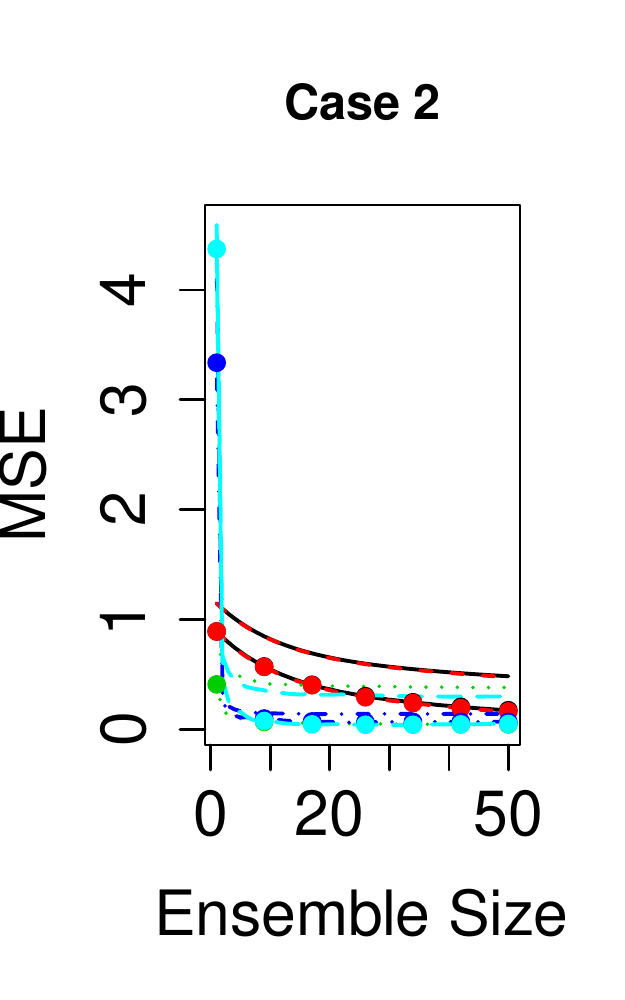}
        \includegraphics[width=0.60\textwidth]{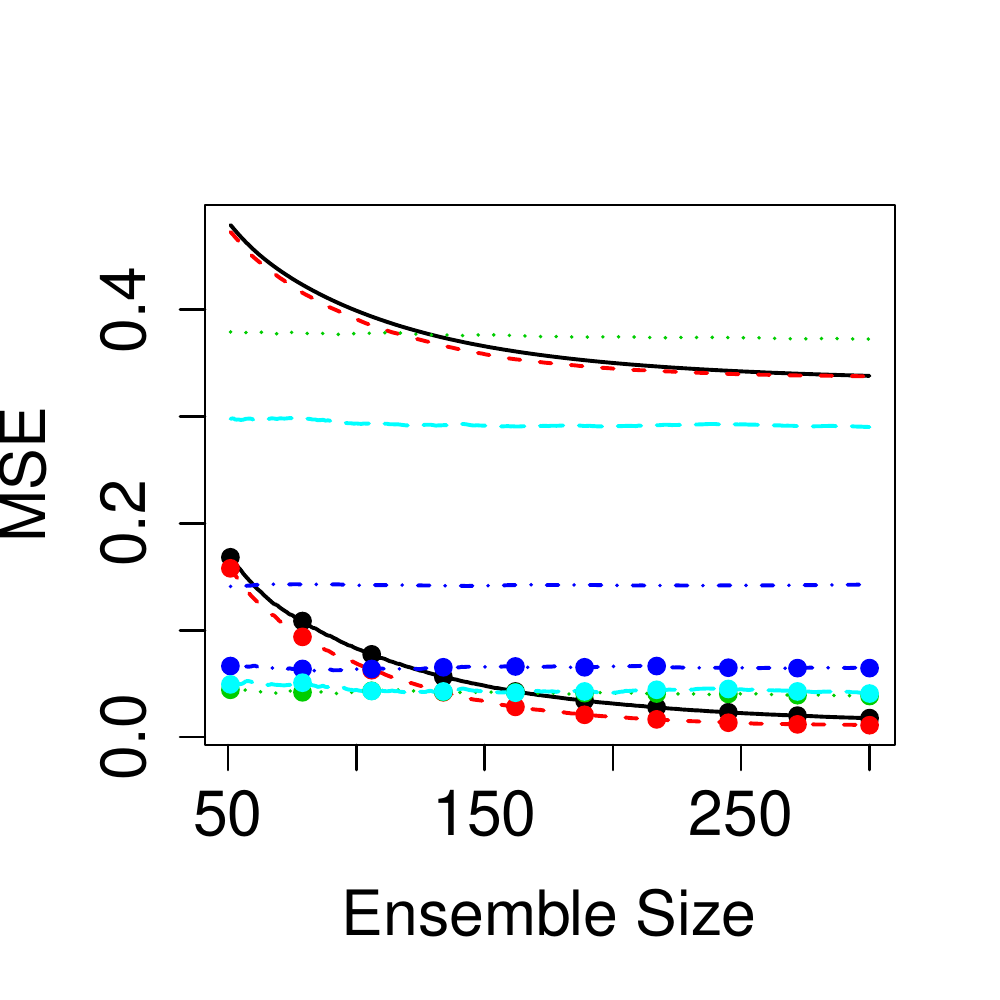}
        \caption{Function (1), Unif\{-1,1\} error.}
   \end{subfigure}
   \includegraphics[width=0.70\textwidth]{synlegend.png}
   \begin{subfigure}[h]{0.49\textwidth}
   	\centering
        \includegraphics[width=0.375\textwidth]{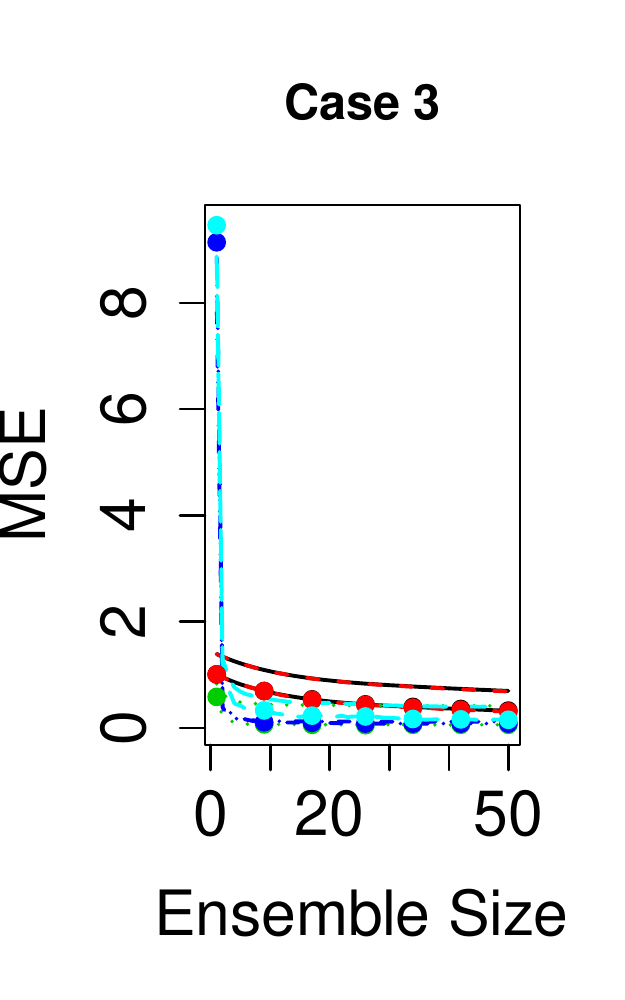}
        \includegraphics[width=0.60\textwidth]{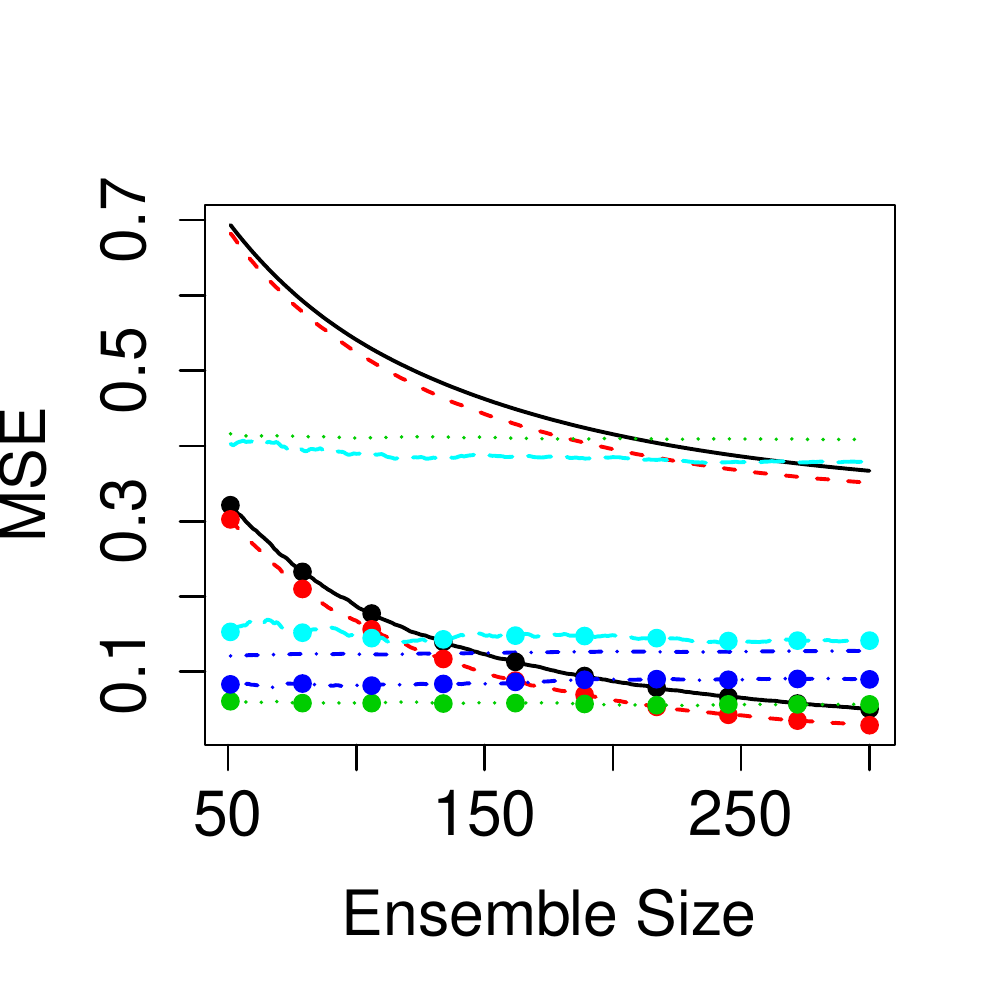}
        \caption{Function (2), Unif[-1,1] error.}
   \end{subfigure}
   \begin{subfigure}[h]{0.49\textwidth}
   	\centering
        \includegraphics[width=0.375\textwidth]{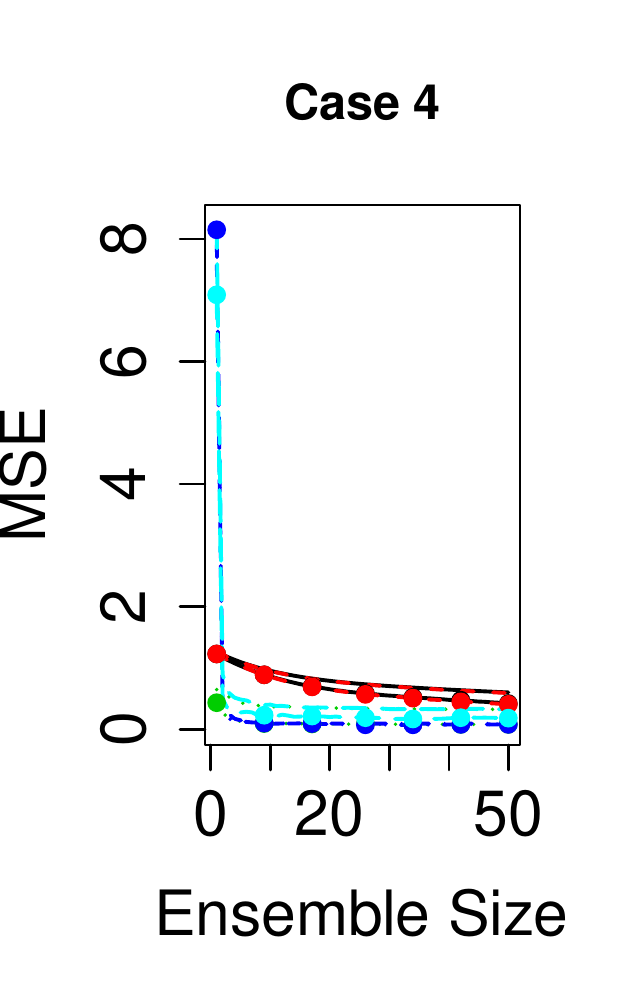}
        \includegraphics[width=0.60\textwidth]{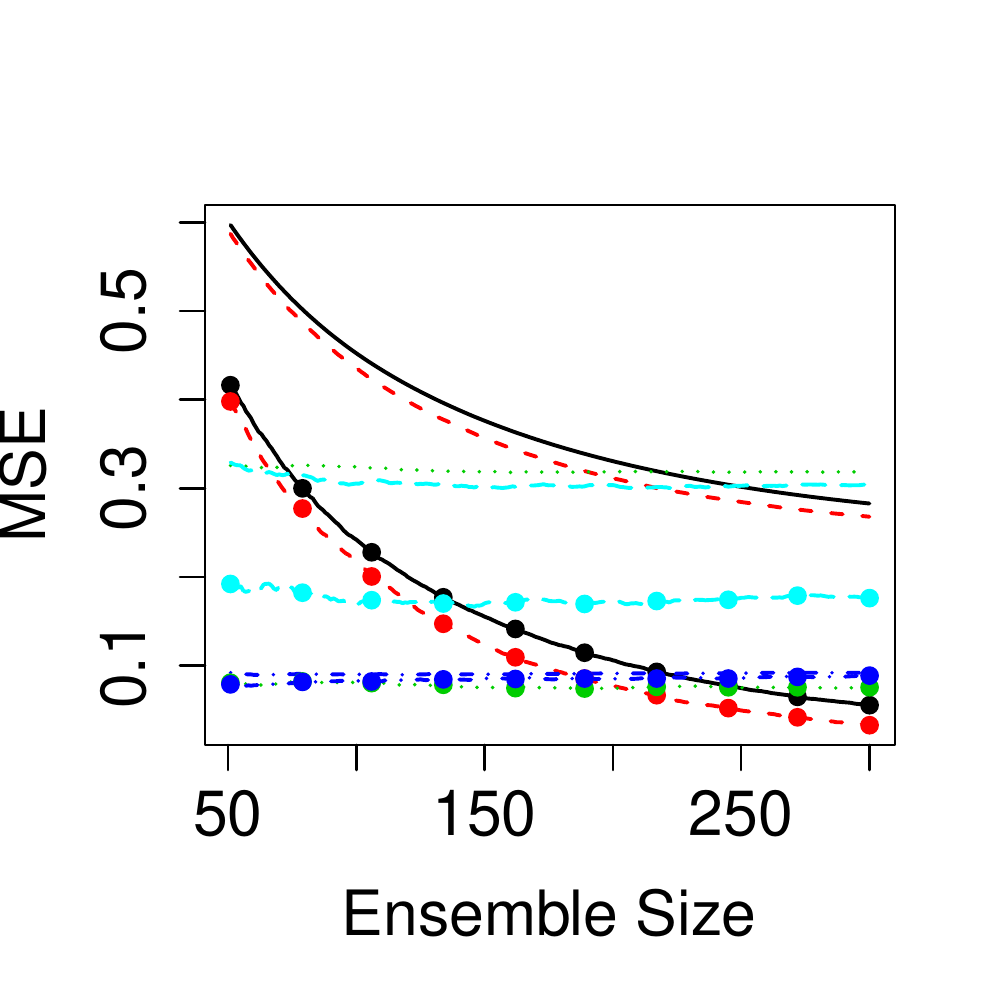}
        \caption{Function (2), Unif\{-1,1\} error.}
   \end{subfigure}
   \caption{Training and testing error curves of tree ensembles on simulated data. \blued{Testing errors are smaller as they are calculated on test sets with no response errors.}}
   \label{fig:sim_syn}
\end{figure}

\blued{
In addition, we have used the same settings and applied different $\lambda$ values 0.2, 0.5 and 0.8 on two simulated data sets to investigate the impact of $\lambda$ on Boulevard accuracy for finite sample cases. Results are summarized in Figure \ref{fig:sim_lambda}. BLV outperforms rBLV while $\lambda=0.8$ appears to be the best choice. Though this observation aligns with our intuition to choose a large $\lambda$, we consider the discrepancy between different $\lambda$ values to be small.
\begin{figure}[h] %  figure placement: here, top, bottom, or page
   \centering
   \includegraphics[width=0.8\textwidth]{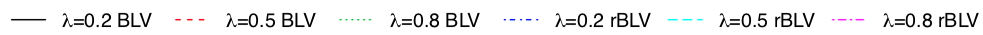}
   \begin{subfigure}[h]{0.49\textwidth}
      \centering
      \includegraphics[width=0.375\textwidth]{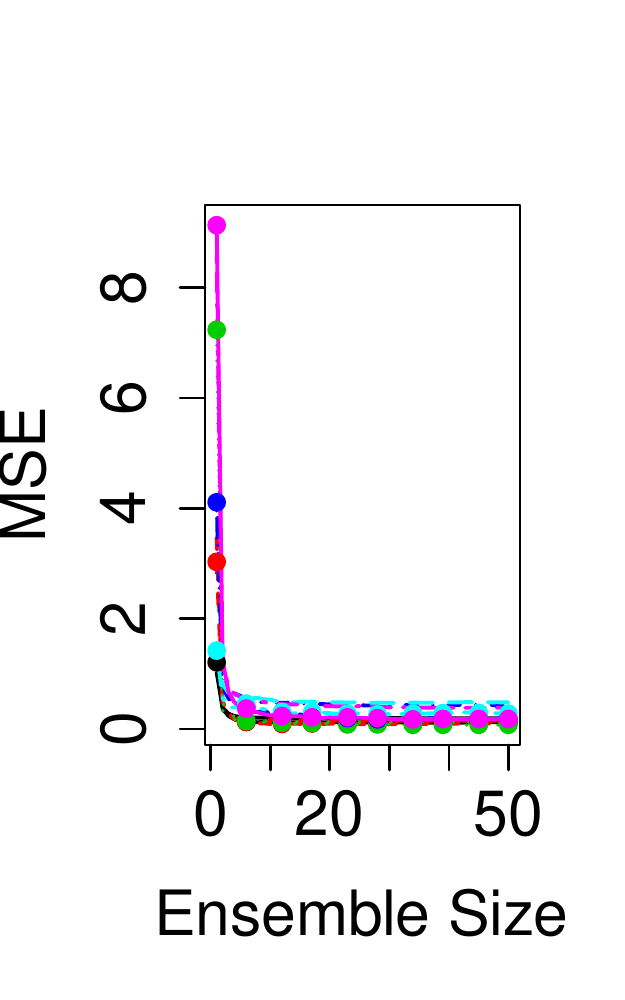}
      \includegraphics[width=0.60\textwidth]{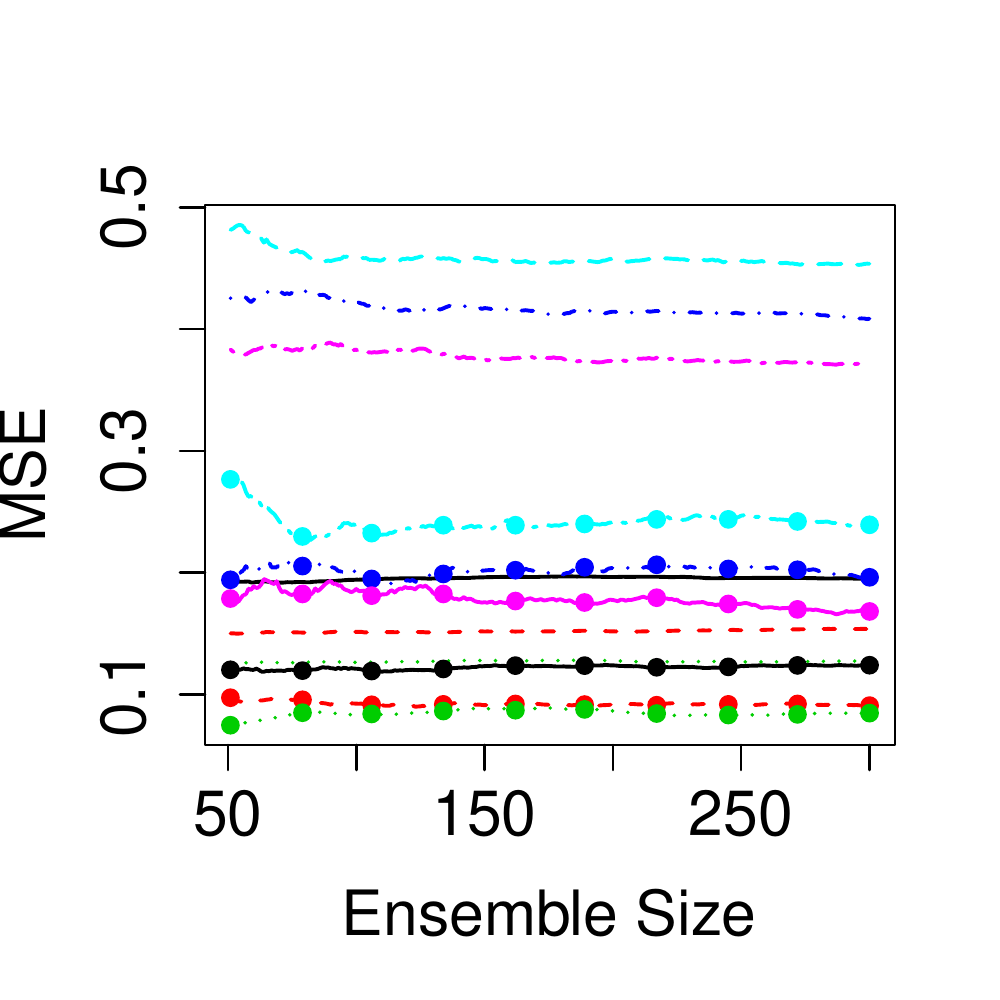}
      \caption{Unif[-1,1] error.}
   \end{subfigure}
   \begin{subfigure}[h]{0.49\textwidth}
      \centering
      \includegraphics[width=0.375\textwidth]{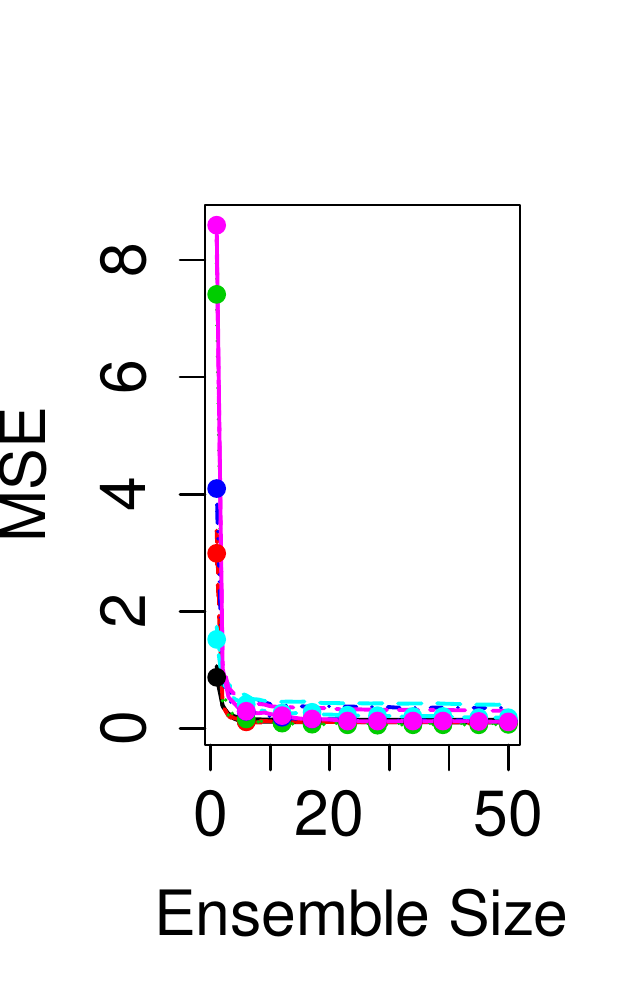}
      \includegraphics[width=0.60\textwidth]{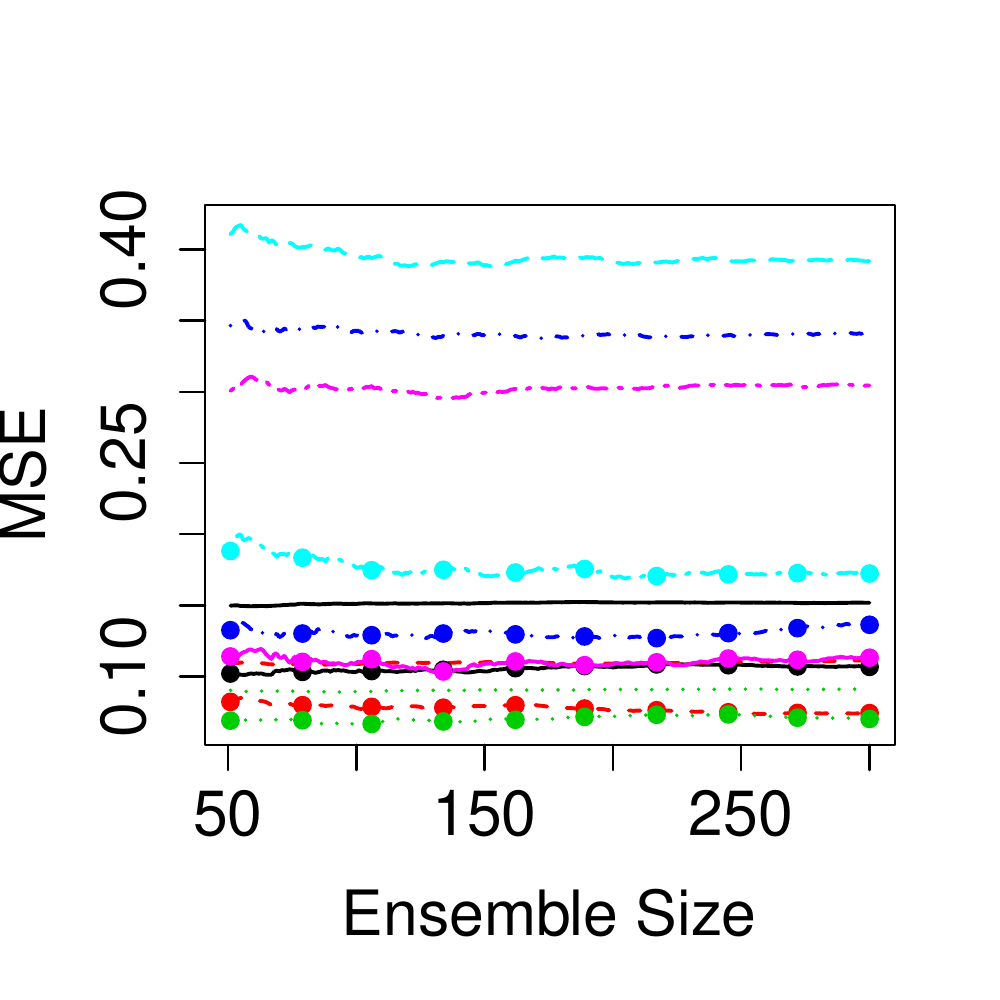}
      \caption{Unif\{-1, 1\} error}
   \end{subfigure}
   \caption{Boulevard accuracy with different choices of $\lambda$ values. Test errors are marked by dots and are smaller as they are calculated on tests set with no response errors.}
   \label{fig:sim_lambda}
\end{figure}
}

Results on four real world data sets \chgx{selected from UCI Machine Learning Repository \citep{Dua:2017, tufekci2014prediction, kaya2012local}} are shown in Figure \ref{fig:sim_real}. All curves are averages after 5-fold cross validation. Different parameters are used for different data sets,  \blued{see Appendix \ref{sec:appemp}}. Rankings of the five methods in comparison \blued{do not show consistency}, nevertheless rBLV and BLV manage to achieve \blued{similar performance to the other methods on test sets}.
\begin{figure}[h] %  figure placement: here, top, bottom, or page
   \centering
   \includegraphics[width=0.70\textwidth]{synlegend.png}
   \begin{subfigure}[h]{0.49\textwidth}
   	\centering
        \includegraphics[width=0.375\textwidth]{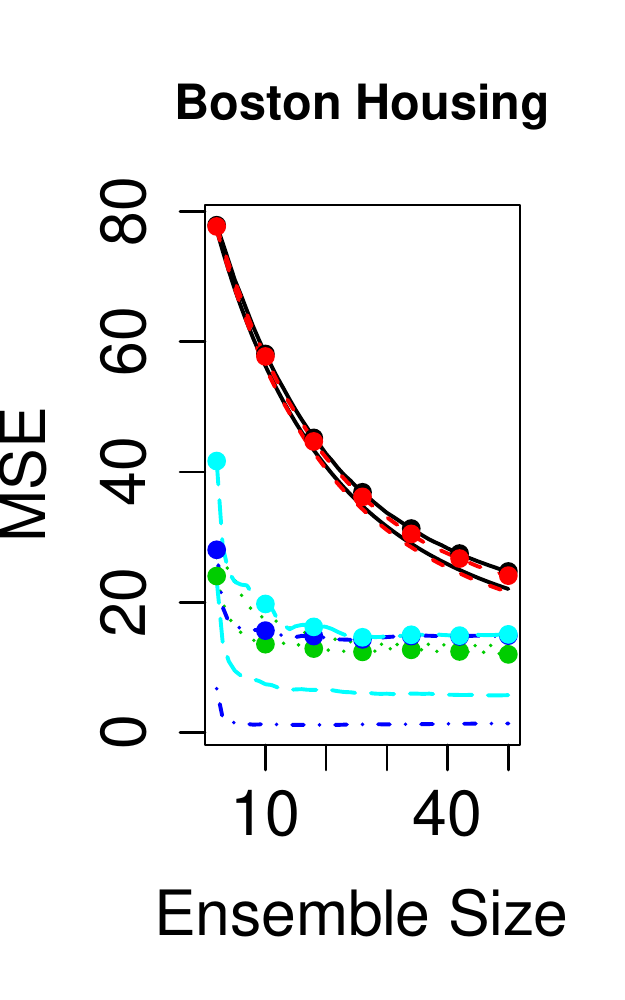}
        \includegraphics[width=0.60\textwidth]{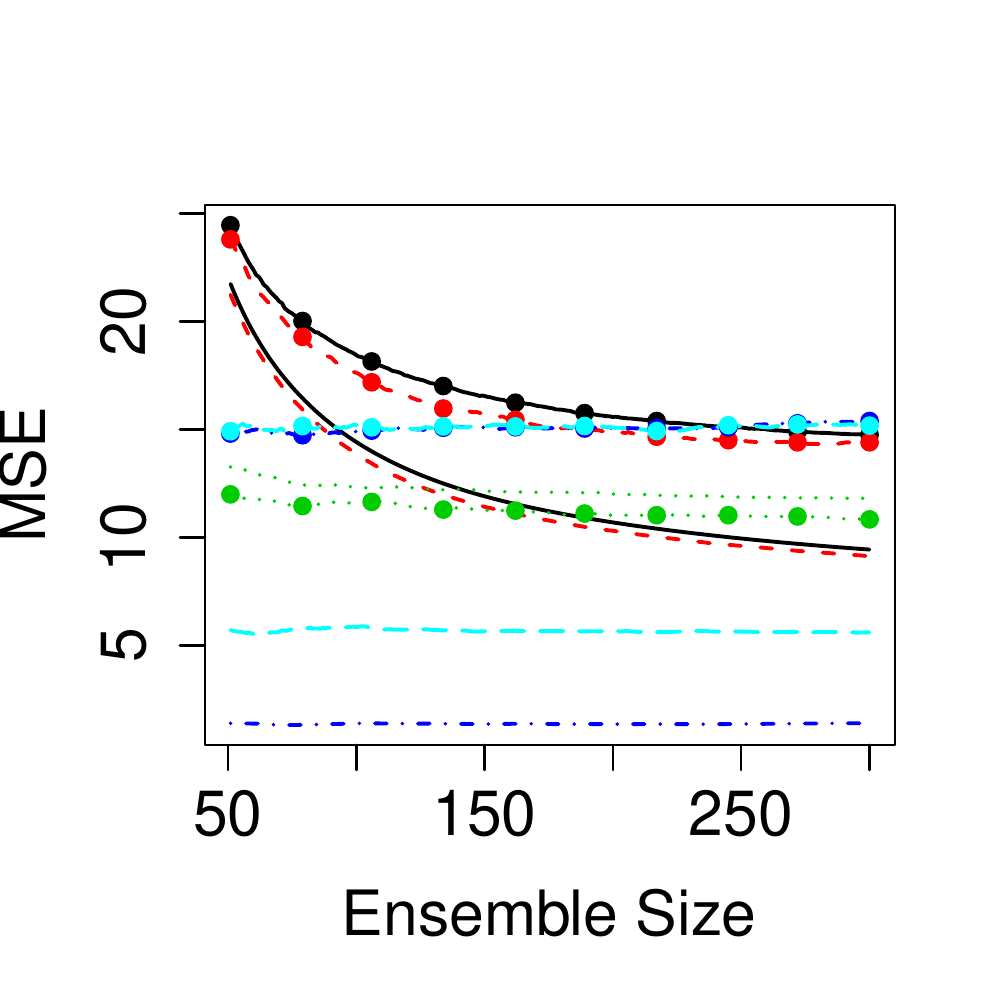}
        \caption{Boston Housing}
   \end{subfigure}
   \begin{subfigure}[h]{0.49\textwidth}
   	\centering
        \includegraphics[width=0.375\textwidth]{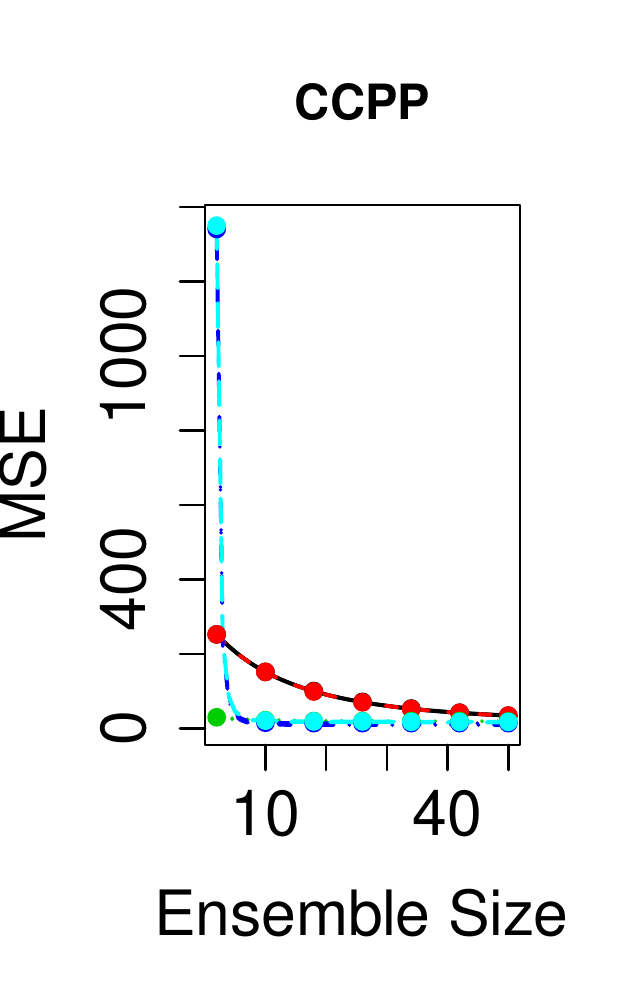}
        \includegraphics[width=0.60\textwidth]{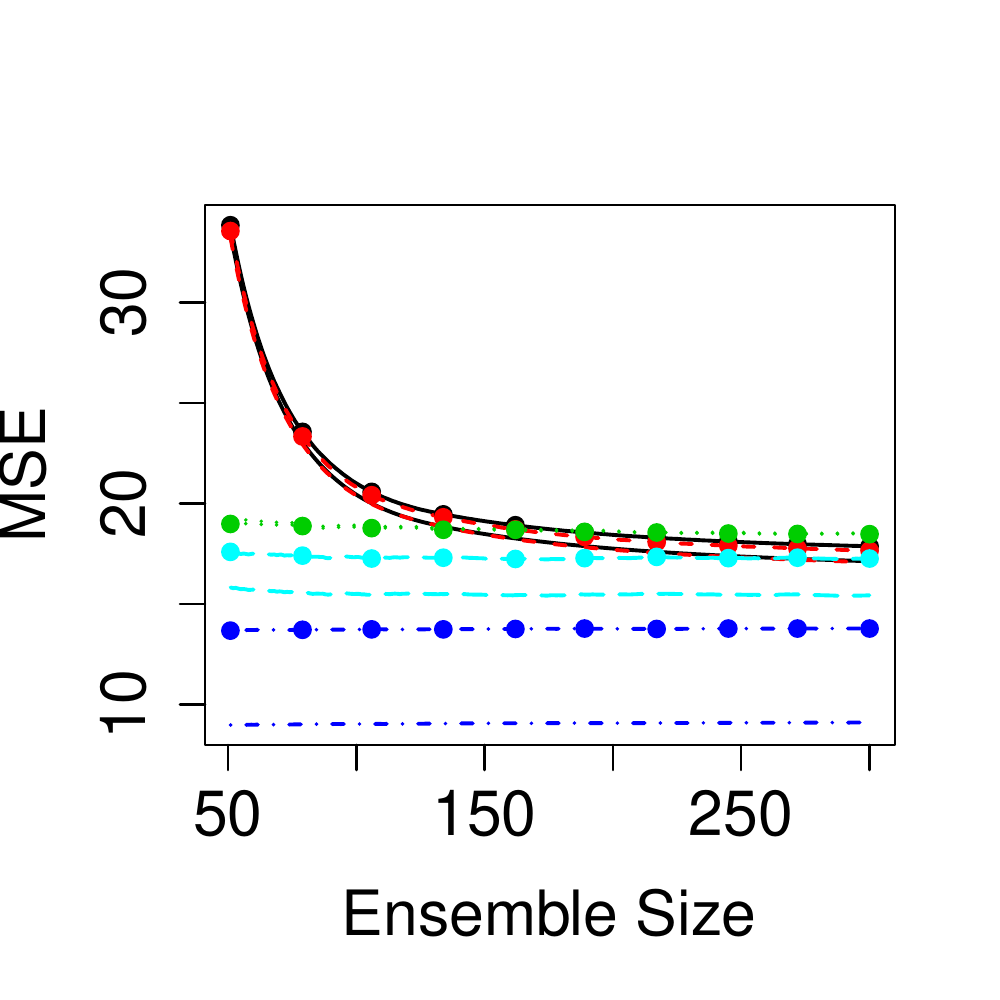}
        \caption{Combined Cycle Power Plant}
   \end{subfigure}
   \includegraphics[width=0.70\textwidth]{synlegend.png}
   \begin{subfigure}[h]{0.49\textwidth}
   	\centering
        \includegraphics[width=0.375\textwidth]{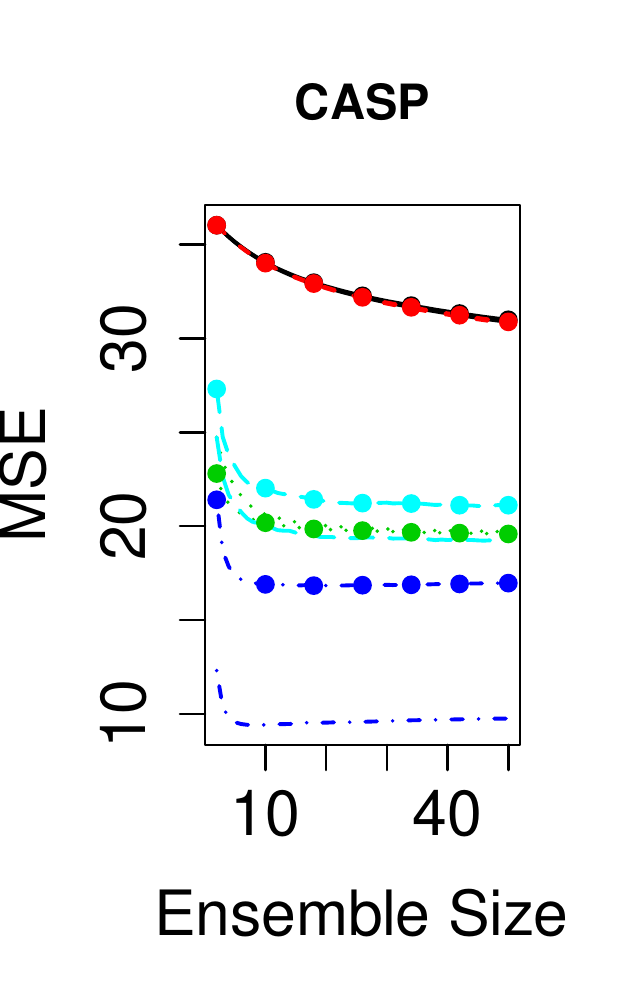}
        \includegraphics[width=0.60\textwidth]{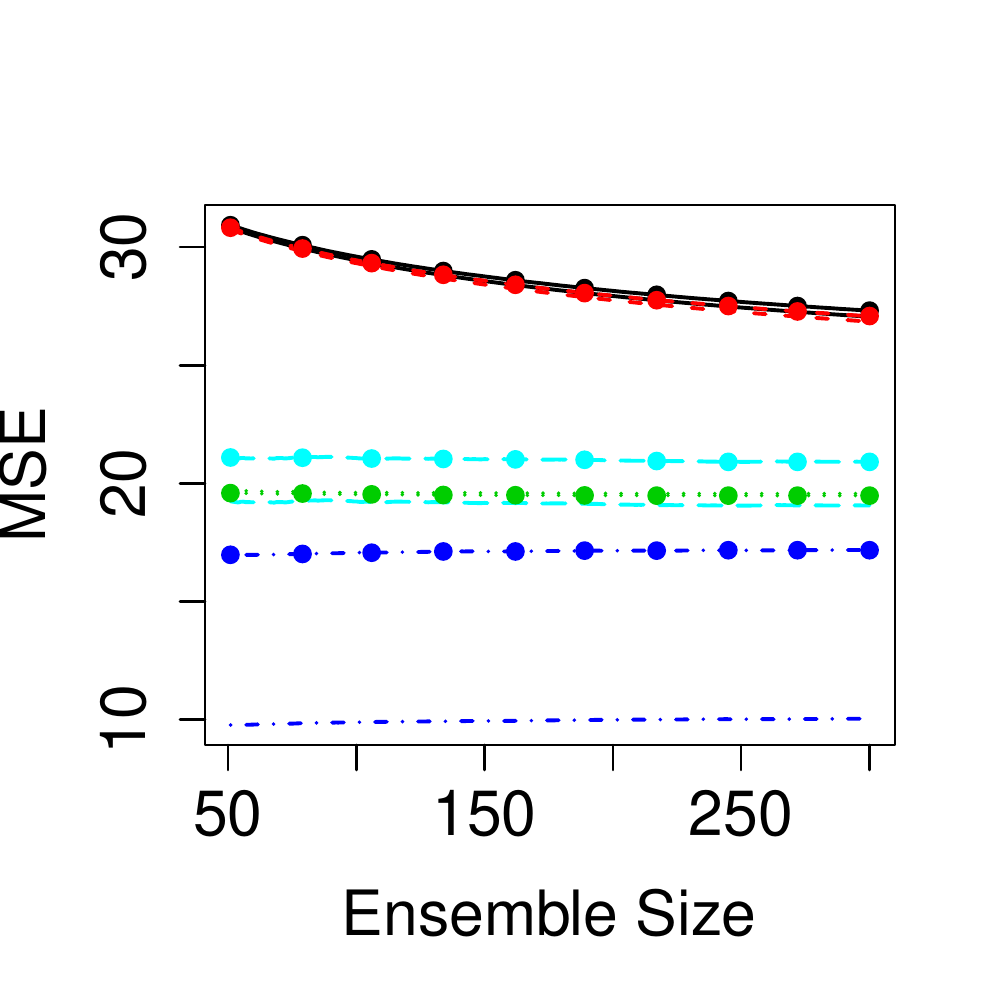}
        \caption{Protein Tertiary Structure}
   \end{subfigure}
   \begin{subfigure}[h]{0.49\textwidth}
   	\centering
        \includegraphics[width=0.375\textwidth]{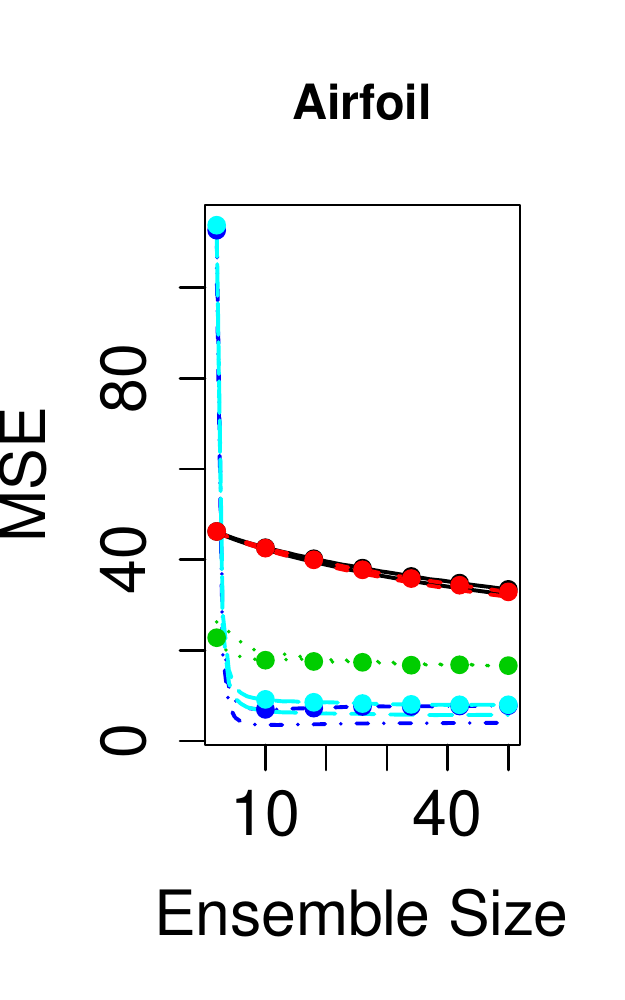}
        \includegraphics[width=0.60\textwidth]{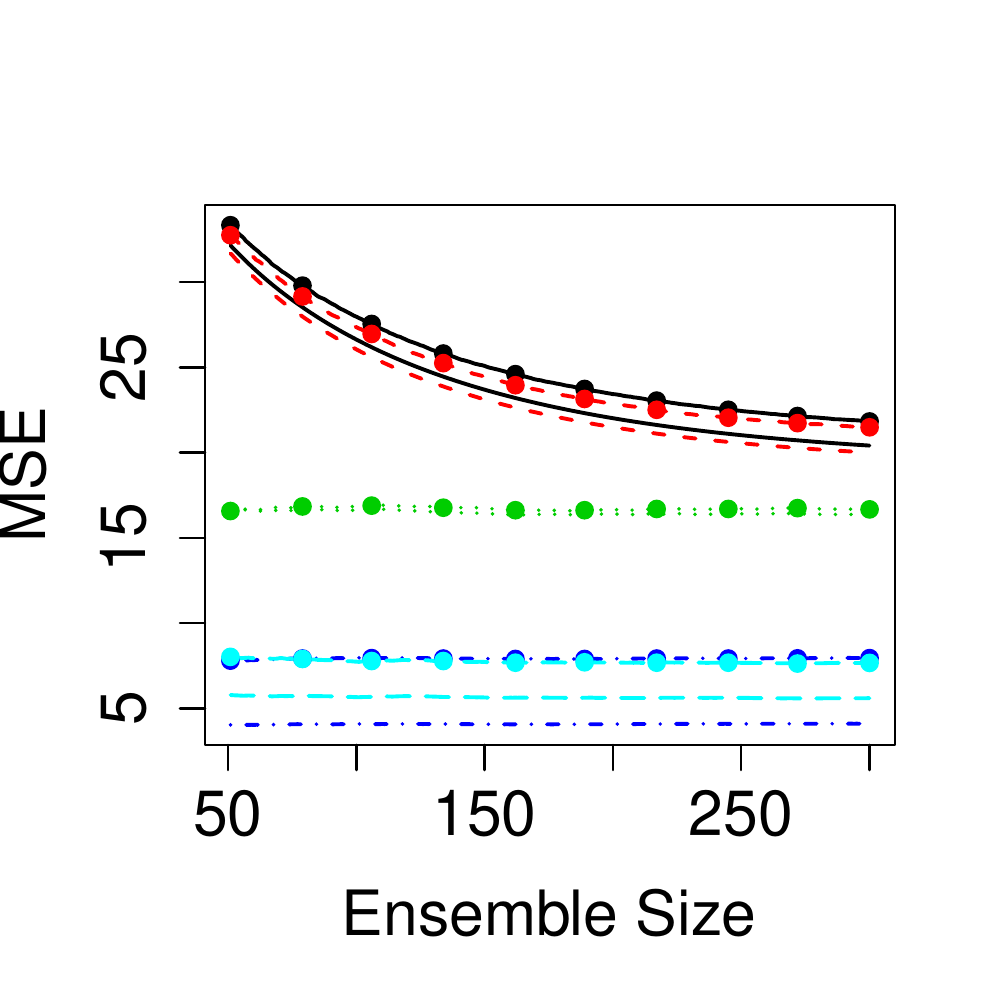}
        \caption{Airfoil}
   \end{subfigure}
   \caption{Training and testing error curves of tree ensembles on real world data sets.}
   \label{fig:sim_real}
\end{figure}

\subsection{Convergence and Kernel Ridge Regression} \blued{
Theorem \ref{thm:boulevardconvergence} and Corollary \ref{cor:prediction} claim that any Boulevard prediction should converge to a kernel ridge regression (KRR) form. Although the matrix inverse involved in the expression is too time consuming to compute for large $n$, it is possible to estimate it using Monte Carlo when $n$ and $d$ are small. We take advantage of this setup to compare Boulevard and KRR in order to empirically support the theorem.

We choose $d=5$ and $n=200$. We use the following model
\begin{align}
\label{fml:sim_mean}
y = x_1 + 3x_2 + x_3^2 + 2x_4x_5+\epsilon, \quad \epsilon \sim Unif[-1,1].
\end{align}
to generate training data. For 4 test points (see Appendix \ref{sec:appemp}), we predict their values using non-adaptive Boulevard with 100 trees, subsampling rate $0.8$ and $\lambda = 0.8$. We also make KRR predictions based on the expected tree structure matrix $\E[S]$ estimated by Monte Carlo using the same 100 trees. This procedure is repeated 20 times with a new training sample each time. Our results are summarized in Figure \ref{fig:blv_kr}.

The plot on the left shows the interim predictions made by one Boulevard run. The shaped dots are Boulevard predictions, while the horizontal lines are the predictions made by KRR. We see that Boulevard iterations converge, while the convergent points agree with KRR predictions. On the right we plot the predictions made by Boulevard against the ones made by KRR for each of the 20 repetitions, with the horizontal and vertical lines marking the true responses. All dots are close to the diagonal $x=y$, indicating these two methods always produce almost identical results. It also shows that the sample variability is influencing the two methods in the same way. Their MSEs, as labeled in the plots, also appear to be indistinguishable.

\begin{figure}[h] %  figure placement: here, top, bottom, or page
   \centering
   \includegraphics[width=0.49\textwidth]{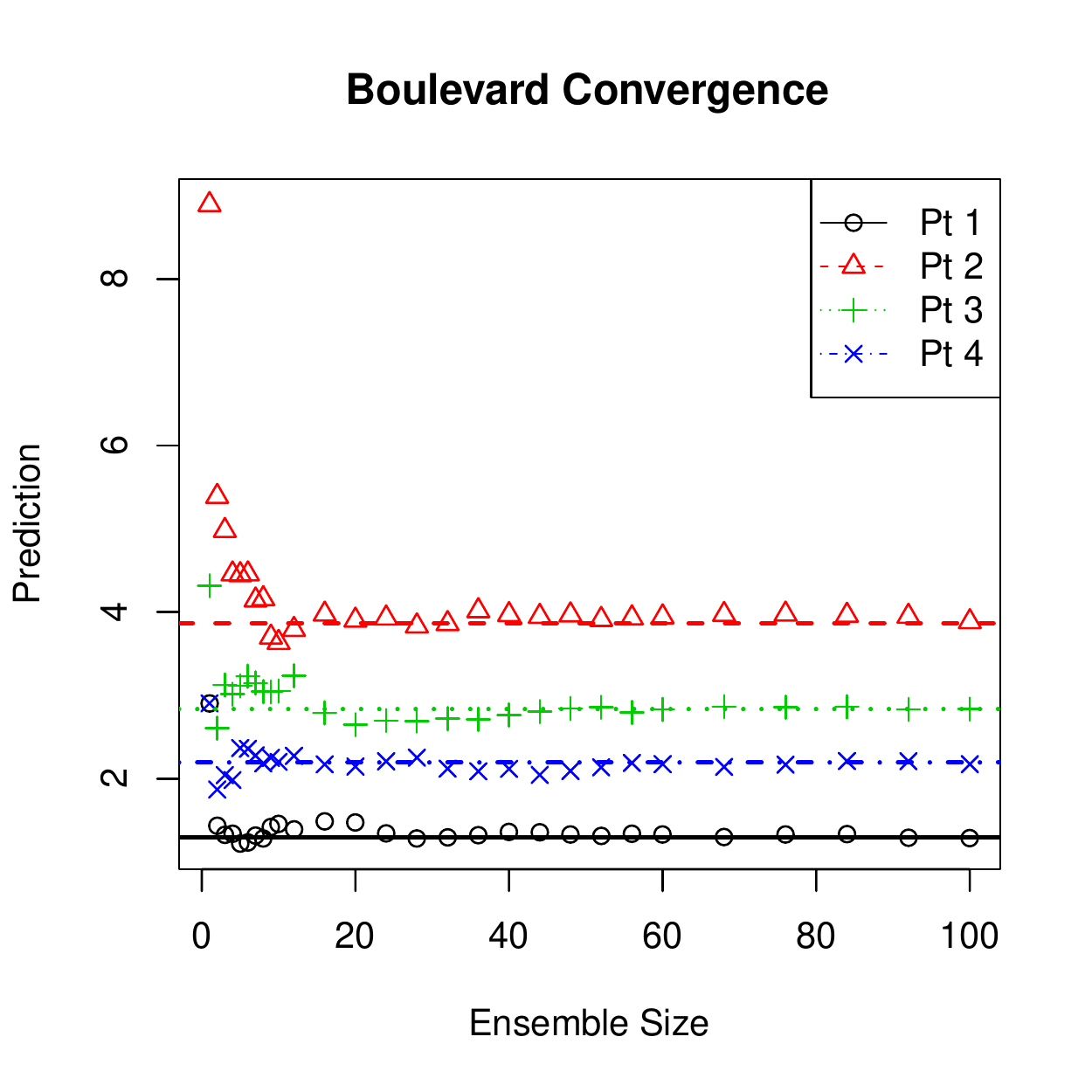}
   \includegraphics[width=0.49\textwidth]{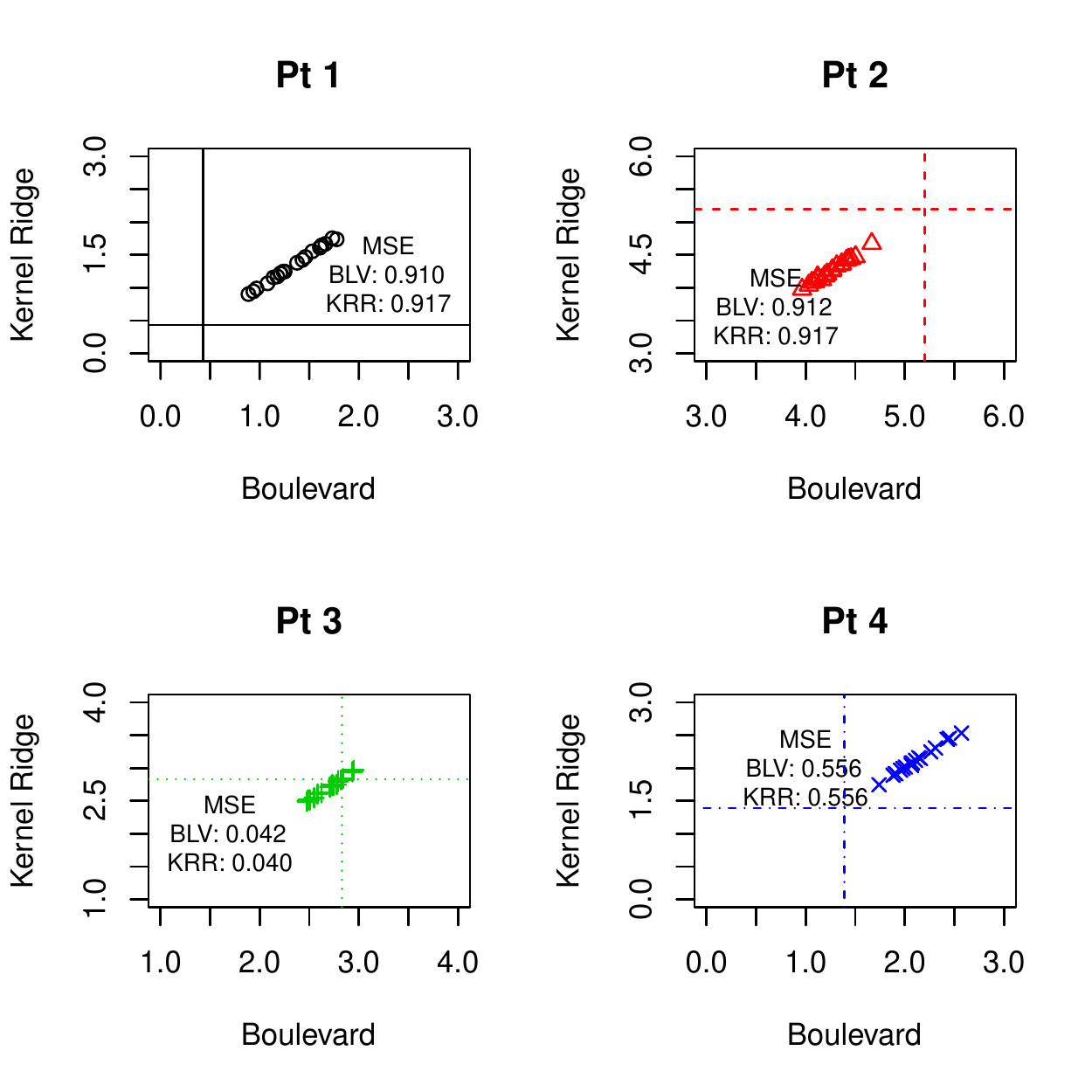}
   \caption{\blued{\textbf{Left}: Boulevard iterations on 4 test points converge to the theoretical convergent points given by the KRR form. The shaped dots are interim Boulevard predictions at different iterations, while horizontal lines are KRR results. \textbf{Right}: Boulevard predictions v.s. KRR predictions. Two methods generate indistinguishable predictions along the diagonal with similar MSEs. True responses are marked by the horizontal and vertical lines.}}
   \label{fig:blv_kr}
\end{figure}
}

\subsection{Limiting Distribution}
To examine the limiting behavior of non-adaptive Boulevard, \blued{we again use the generative model}
\begin{align*}
y = x_1 + 3x_2 + x_3^2 + 2x_4x_5 + \epsilon.
\end{align*}
A set of 10 fixed test points \chgx{(see Appendix \ref{sec:appemp}) }are used along the experiments. We \changed{set a sample size of 1000}, add different sub-Gaussian error terms to \changed{this signal and} build non-adaptive Boulevard until ensemble size reaches 2000. This is repeated 1000 times with a new sample each time and we plot the distribution of the predictions\deleted{ as shown} in Figure \ref{fig:pred_dist}. \deleted{Sample size is fixed at 1000. }All these curves are undistinguishable from normal distribution by Kolmogorov-Smirnov test.

\begin{figure}[h] %  figure placement: here, top, bottom, or page
   \centering
   \includegraphics[width=0.49\textwidth]{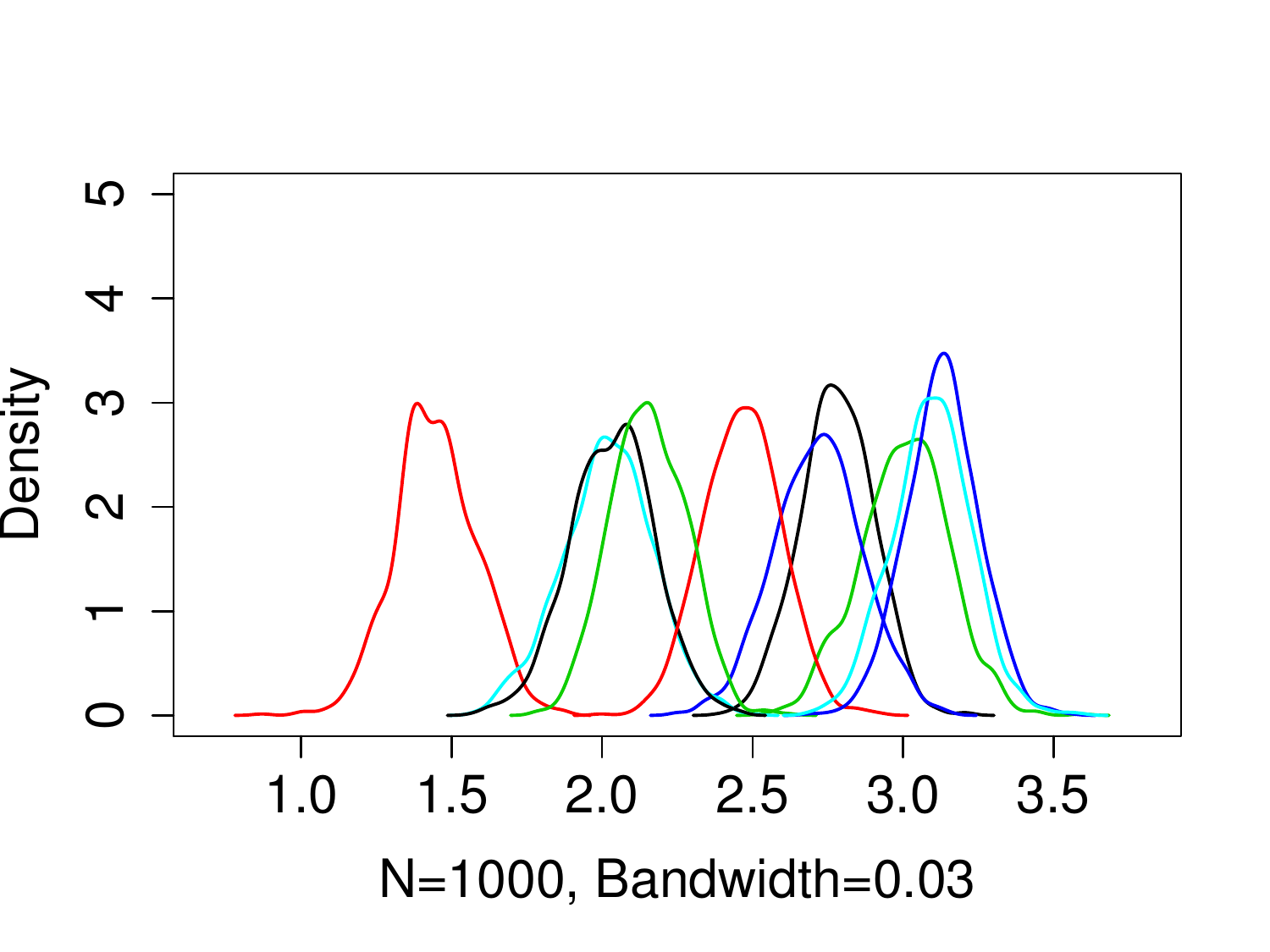}
   \includegraphics[width=0.49\textwidth]{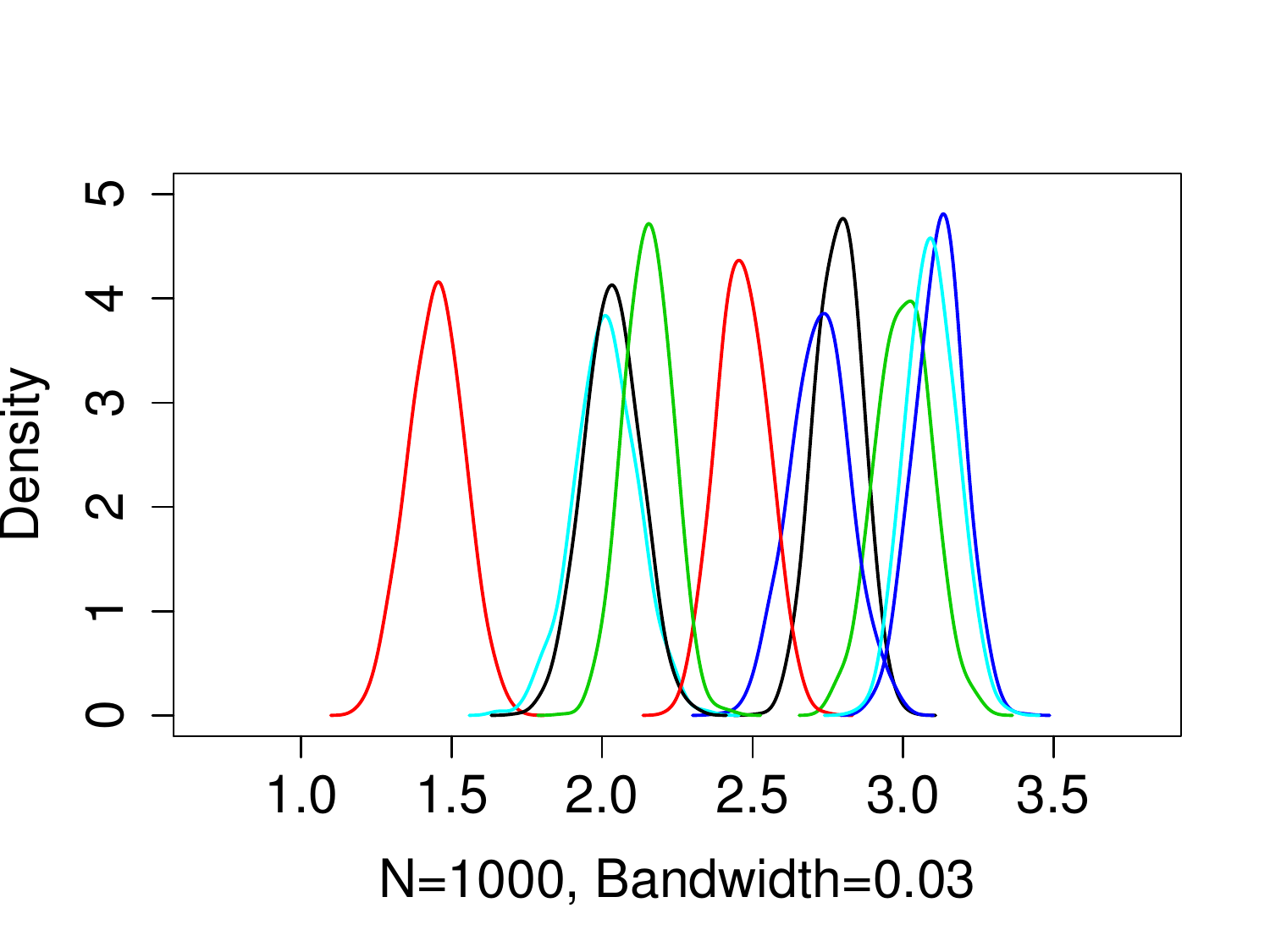} \\
   \includegraphics[width=0.49\textwidth]{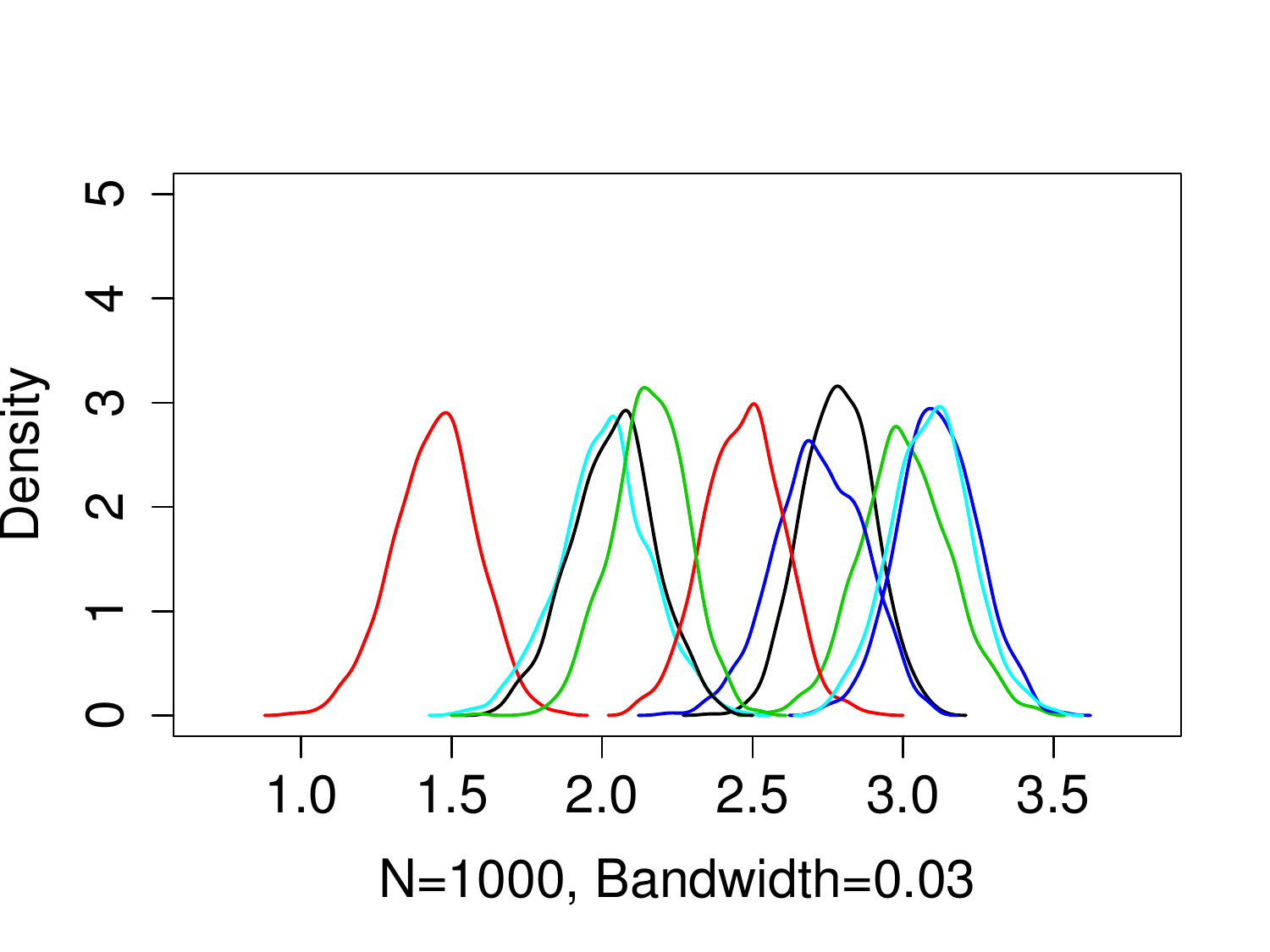}
   \includegraphics[width=0.49\textwidth]{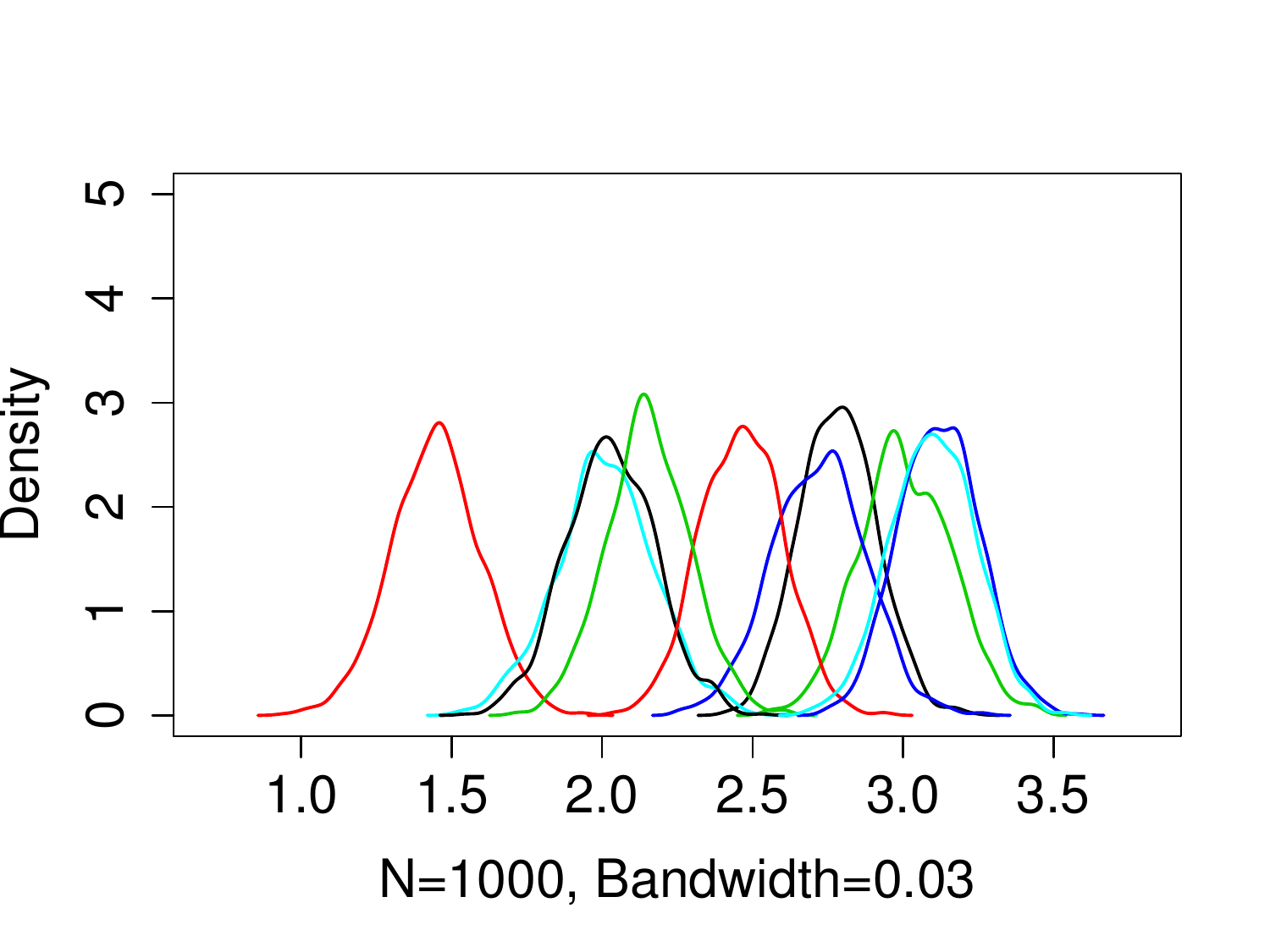}
   \caption{Distributions of predictions of test points with different error terms. The errors are N(0,1), Unif[-1,1], Unif$\{-1, 1\}$, and half chance -1 half chance Unif[0,2], respectively. }
   \label{fig:pred_dist}
\end{figure}

\subsection{Reproduction Interval}
Similar to prediction intervals which quantify the uncertainty of future predictions, we introduce the {\em reproduction interval} as the uncertainty measure for where the prediction would be if it were made on another independent sample. Theorem \ref{thm:main} is used to create reproduction intervals for Boulevard. $k_n$ in the stochastic variance is empirically estimated directly using the ensemble, while $[\frac{1}{\lambda} I + K_n]^{-1}$ is \deleted{simplified to $\lambda$ by its norm conservatively}\changed{conservatively simplified to its largest possible norm $\lambda$}. We then scale the variance estimate by 2 to \changed{account for having separate independent samples}: \blued{With two independent samples, we expect the prediction on the same point $x$ to be $\hat{f}_1(x) \sim N(f(x), \sigma_x^2)$ and $\hat{f}_2(x)\sim N(f(x), \sigma_x^2)$ respectively, implying $\hat{f}_1(x) - \hat{f}_2(x) \sim N(0, 2\sigma_x^2)$.}

We \changed{use the training sample to create reproduction intervals for the test points, then repeatedly train and} predict each test point \changed{for another} 100 times with a different sample each time. Figure \ref{fig:pi} shows the 95\% reproduction intervals we capture under different settings. We anticipate more accurate results with larger sample size.

\begin{figure}[h] %  figure placement: here, top, bottom, or page
   \centering
   \includegraphics[width=0.49\textwidth]{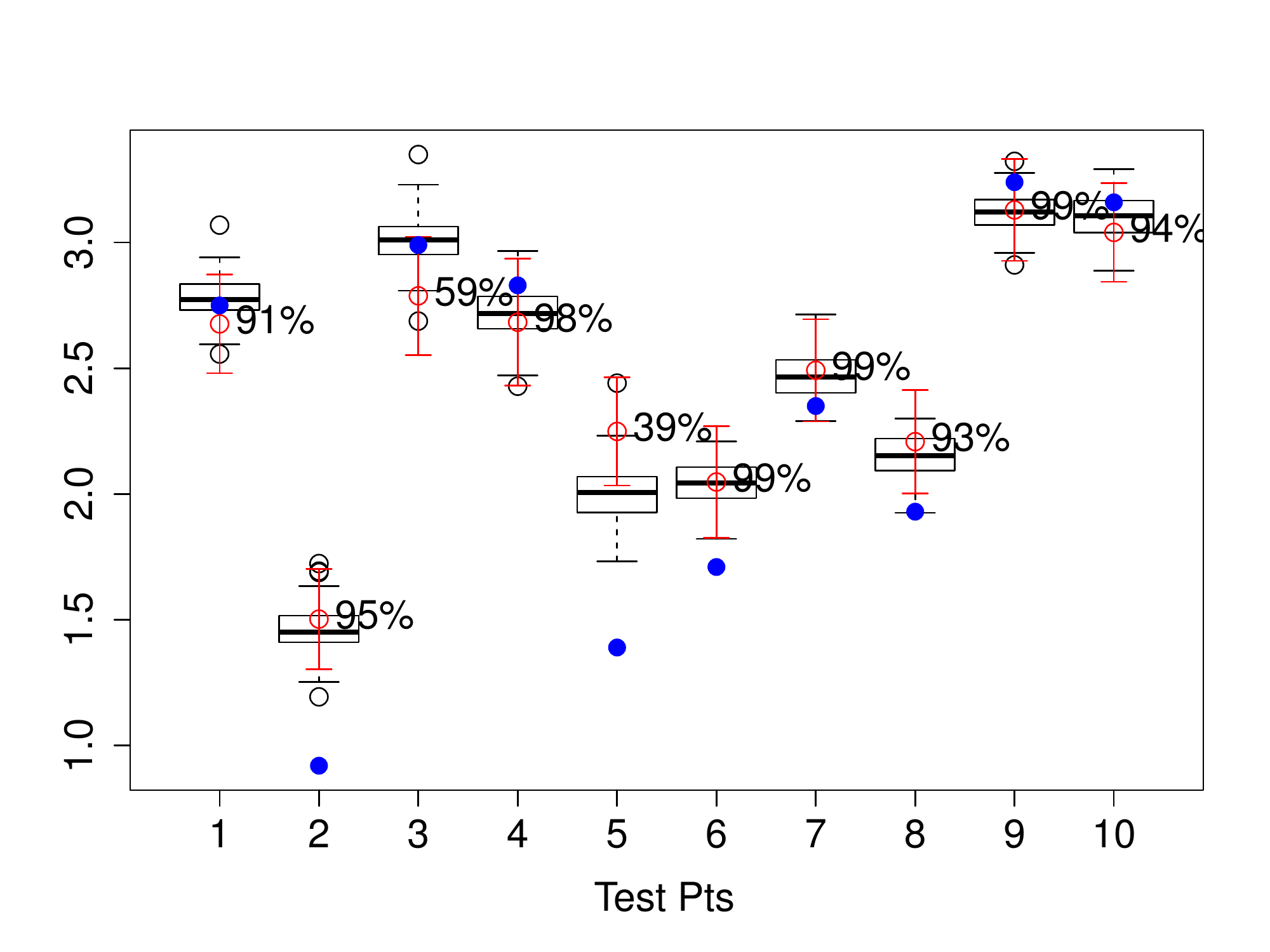}
   \includegraphics[width=0.49\textwidth]{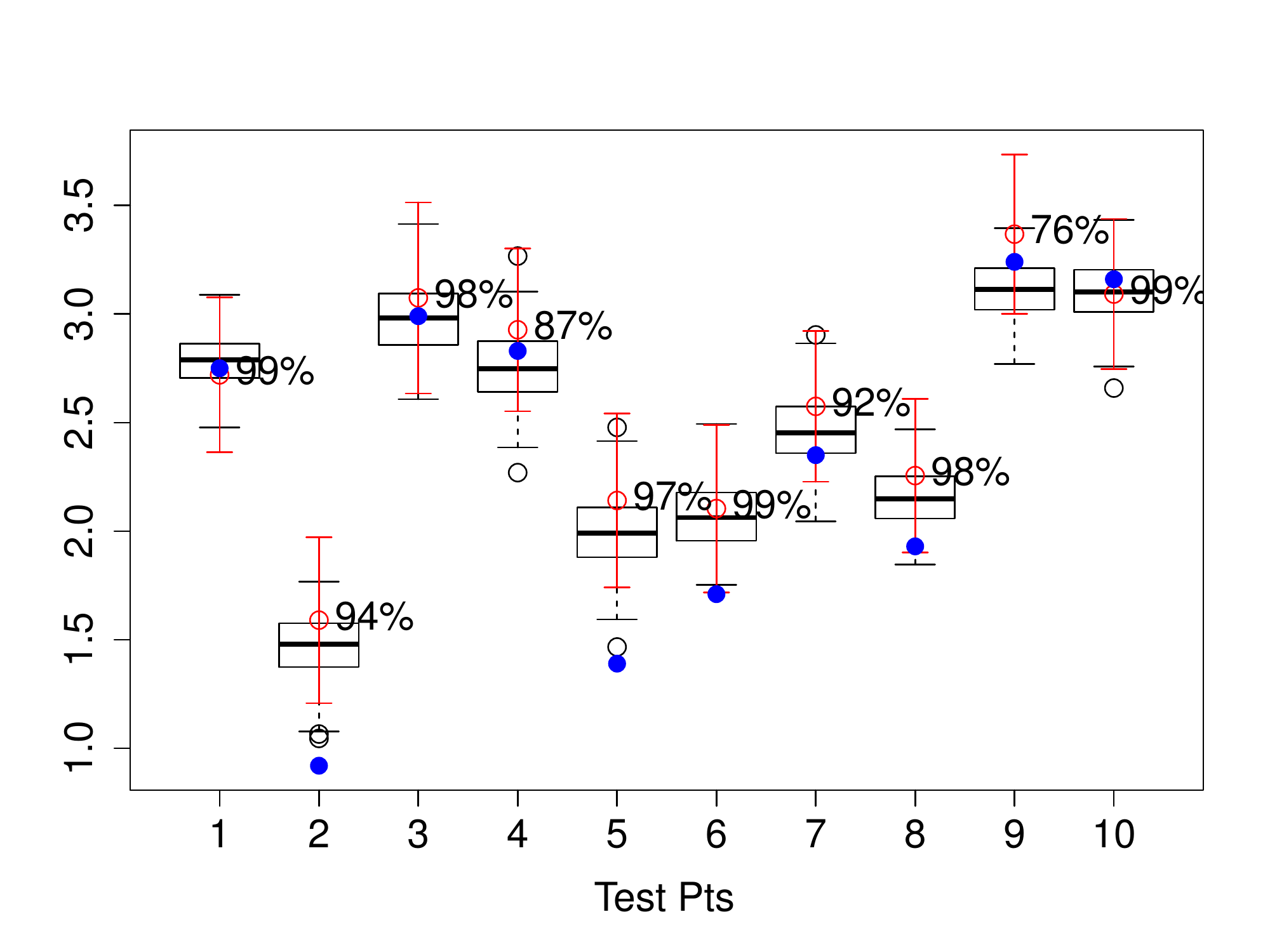} \\
   \includegraphics[width=0.49\textwidth]{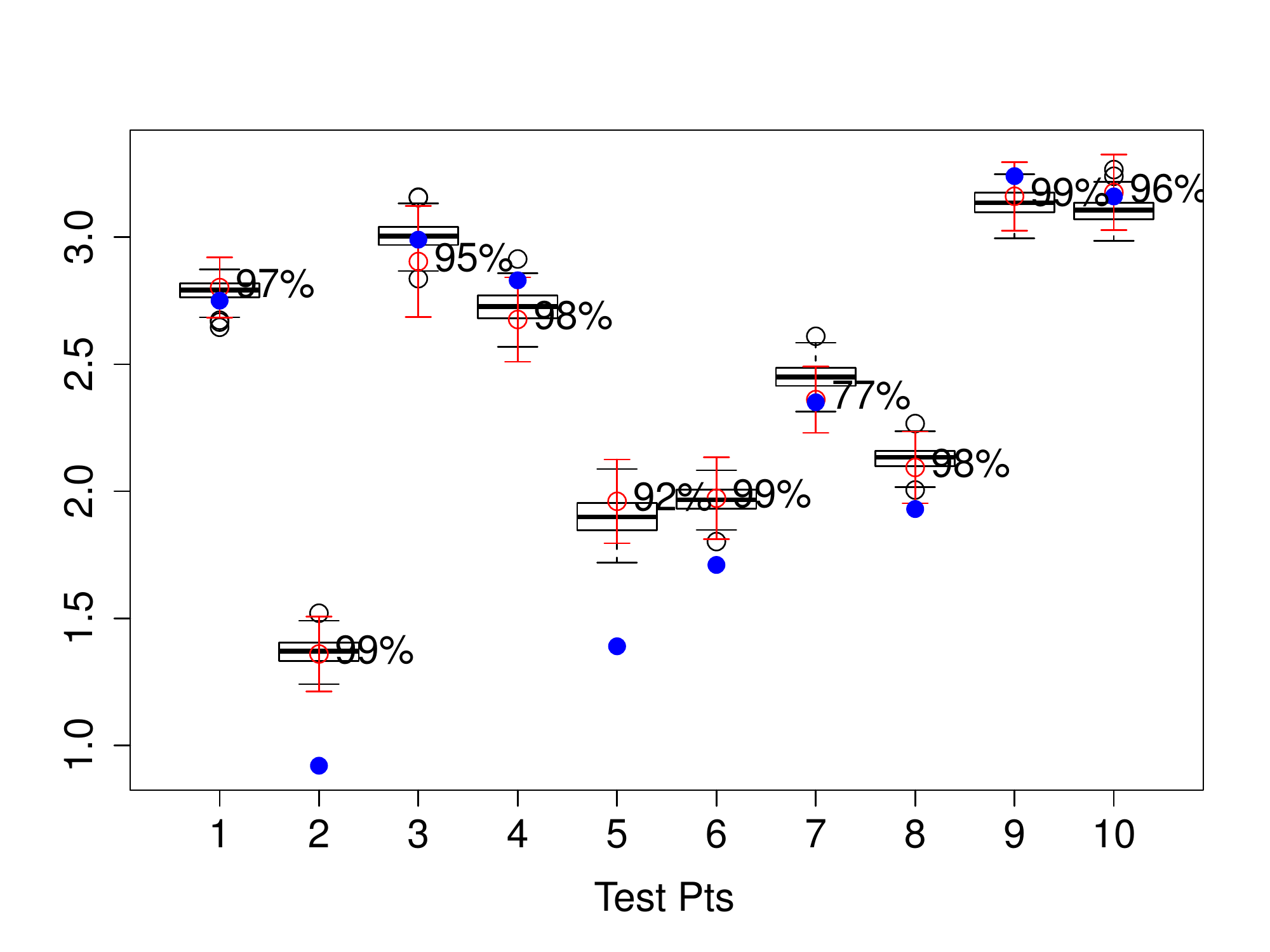}
   \includegraphics[width=0.49\textwidth]{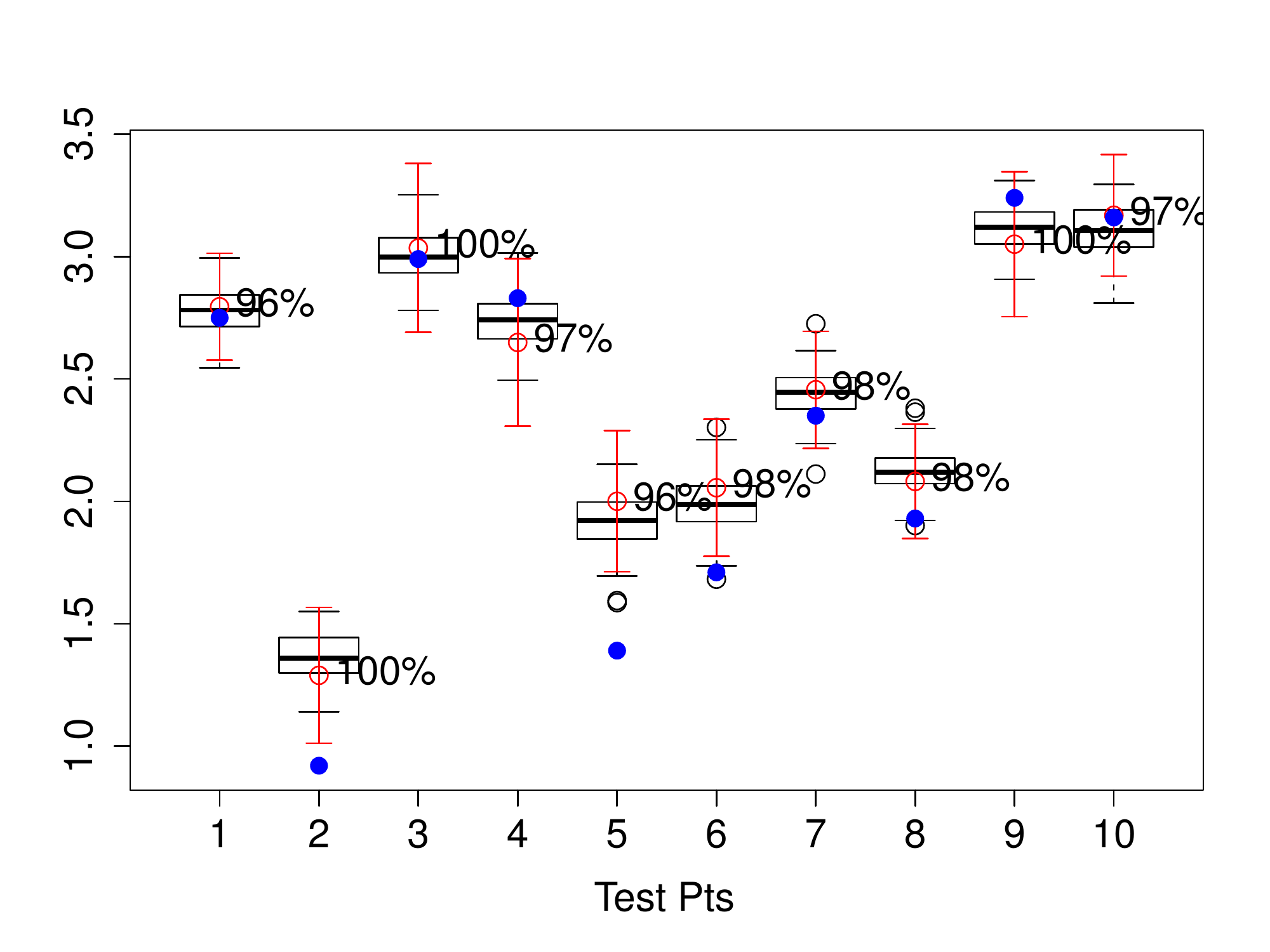}
   \caption{Reproduction intervals. Boxplots show distributions of predictions; red intervals are reproduction intervals; blue dots are truths. Sample sizes are 1000 (top row) and 5000 (bottom row), error terms Unif[-1,1] (left column) and Unif[-2, 2] (right column). Coverage is shown by numbers next to interval centers. }
   \label{fig:pi}
\end{figure}

Furthermore, we notice the uniform pattern of biases in those plots. \changed{This bias comes from }two known causes. One is that \changed{we are using small samples} which are far from guaranteeing the consistency. The other is because of the edge effects; the distance of the ten chosen test points to the center of the hypercube is respectively 0.000, 0.671, 0.894, 0.894, 0.894, 0.693, 0.520, 0.436, 0.510 and 0.469. We in general expect biased prediction when the point is near the boundary.

\section{Discussion}
This work is so far the first we know to have \deleted{stablised statements of GBT limiting distributions}\changed{established a limiting distribution for gradient boosted trees}. The roadmap consists of the following key components. We implemented the honest tree construction to reduce the chance of chasing \changed{order statistics}, applied downweighting towards averaging to achieve convergence, and carefully selected tree construction rate to obtain asymptotic normality. With the uncertainty measure of Boulevard predictions, we are looking forward to exploring the use of regularized gradient boosting \changed{within inference}, making the model more interpretable and analyzable.

The sequential correlation \changed{induced by GBT} is the major issue complicating our analysis\deleted{ which prevents the easy approaches to asymptotic normality}, while the resistance of decision trees to mathematical quantification adds more complication. Much of our effort is spent on seeking rational conditions to compromise these features. Structure-value isolation, non-adaptivity and down-weighing all contribute to making sequential trees less correlated.

As briefly discussed above, there are two remaining questions in our paper. The first question regards the stochastic variance. To allow a richer collection of decision tree construction strategies, it is essential to discover the weakest tree condition under which the variance term converges in probability. The second question is the convergence \deleted{conclusion on}\changed{of} adaptive Boulevard which appears to \deleted{be more feasible}\changed{hold} in practice. In spite of the \textit{ad hoc} tail snapshot we proposed, the existence and uniqueness of the contraction region of generic Boulevard requires more mathematical formulation of decision trees or\deleted{, possibly,} some variation of decision trees.

\bibliographystyle{chicago}
\bibliography{main}

\appendix

\section{Proofs}
%%%% SEPARATE %%%%
\subsection{Proof of Theorem \ref{thm:treeaskernel}}
\begin{proof}
\blued{We will prove the theorem with respect to $\E_w[\cdot]$. The $\E_{w,q}[\cdot]$ conclusion is an instant corollary as the three properties are additive.}
To prove (1), element-wise non-negativity is trivial. To show symmetry, consider any given $i\not=j$ and assume $x_i \in A$ and $x_j \in A'$ under the assumption of subsample uniformity,
\begin{align*}
\E_w [S_n]_{i,j} = \E_w [s_{n,j}(x_i)] & = \frac{1}{\binom{n}{\theta n}} \sum_w \frac{I(x_j \in A)I(j \in w)}{\sum_{x_l \in A} I(l \in w)} \\
\E_w [S_n]_{j,i} = \E_w [s_{n,i}(x_j)] & = \frac{1}{\binom{n}{\theta n}} \sum_w \frac{I(x_i \in A')I(i \in w)}{\sum_{x_l \in A'} I(l \in w)}.
\end{align*}
Therefore $\E_w[S_n]_{i,j} = \E_w[S_n]_{j,i} = 0$ if $A \not= A'$.

In the cases of $A = A'$, $I(x_j \in A) = I(x_i \in A') = 1$. We consider the following possibilities of $w$. \\(a) For $i \notin w, j \notin w$,
$$
\frac{I(j \in w)}{\sum_{x_l \in A} I(l \in w)} = \frac{I(i \in w)}{\sum_{x_l \in A} I(l \in w)} = 0.
$$
(b) For $i \in w, j \in w$,
$$
\frac{I(j \in w)}{\sum_{x_l \in A} I(l \in w)} = \frac{I(i \in w)}{\sum_{x_l \in A} I(l \in w)} = \frac{1}{\sum_{x_l \in A} I(l \in w)}.
$$
(c) For $i \in w, j \notin w$, consider $w' = w  \backslash \{i\} \cup \{j\}$ s.t. $\sum_{x_l \in A} I(l \in w) = \sum_{x_l \in A} I(l \in w')$,
$$
\frac{I(j \in w')}{\sum_{x_l \in A} I(l \in w')} = \frac{I(i \in w)}{\sum_{x_l \in A} I(l \in w)} = \frac{1}{\sum_{x_l \in A} I(l \in w)}.
$$
(d) Similarly, for $i \notin w, j \in w$, consider $w' = w  \backslash \{j\} \cup \{i\}$,
$$
\frac{I(j \in w)}{\sum_{x_l \in A} I(l \in w)} = \frac{I(i \in w')}{\sum_{x_l \in A} I(l \in w')} = \frac{1}{\sum_{x_l \in A} I(l \in w)}.
$$
Since all $w$'s are equally likely, we conclude by symmetry that $\E_w [S_n]_{i,j} = \E_w [S_n]_{j,i}$, hence $\E_w [S_n]$ is symmetric.

To prove (2), notice $\forall x_i, x_j, x_k \in A$,
$$
\E_w[S_n]_{k,i} = \frac{1}{\binom{n}{\theta n}} \sum_w \frac{I(i \in w)}{\sum_{x_l \in A} I(l \in w)} = \E_w[S_n]_{j,i}.
$$
Therefore $\E_w[S_n]$, after proper permutation to gather points in same leaves together, is diagonally blocked with equal entries in each diagonal block and 0 elsewhere, thus positive semi-definite.

To show (3), notice that $S_n$ has column sums of $\leq1$ (not $=1$ due to chances of having no subsample points in the leaf), so does $\E_w[S_n]$. Thus $\norm{\E_w[S_n]}_1\leq 1$. Similarly, $\E_w[S_n]$ has rows sums of $\leq 1$ due to its symmetry therefore $\norm{\E_w[S_n]}_{\infty} \leq 1$. Therefore,
\begin{align*}
\rho(\E_w[S_n]) = \norm{\E_w[S_n]} & \leq \sqrt{\sup_{\norm{u} = 1} \E_w[S_n]^T \E_w[S_n] u} \\
& \leq  \sqrt{\sup_{\norm{u} = 1} \norm{\E_w[S_n]^T}_{\infty} \norm{\E_w[S_n]}_{\infty} \norm{u}_\infty} \\
& \leq  \sqrt{\sup_{\norm{u} = 1} \norm{\E_w[S_n]}_{1} \norm{\E_w[S_n]}_{\infty} \cdot 1} \\
& \leq 1.
\end{align*}
\end{proof}

%%%% SEPARATE %%%%
\subsection{Stochastic Contraction}
\label{sec:app:sc}
\begin{definition}[Stochastic Contraction] Given real-valued stochastic process $\{X_t\}_{t \in \mathbb{N}}$, a sequence of $0 < \lambda_t \leq 1$, define
\begin{gather*}
\mathcal{F}_0 = \emptyset, \mathcal{F}_t = \sigma(X_1,\dots, X_t), \\
\epsilon_t =  X_t - \E[X_t | \mathcal{F}_{t-1}].
\end{gather*}
We call $X_t$ a stochastic contraction if the following is satisfied
\begin{itemize}
\item Vanishing coefficients $$\sum_{t=1}^{\infty} (1-\lambda_t) = \infty, \mbox{ i.e. } \prod_{t=1}^{\infty}\lambda_t = 0.$$
\item Mean contraction $$\lambda_tX_{t-1}I(X_{t-1}\leq 0) \leq \E[X_t | \mathcal{F}_{t-1}] \leq \lambda_{t} X_{t-1}I(X_{t-1}\geq 0), a.s..$$
\item Bounded deviation $$\sup |\epsilon_t| \to 0, \quad \sum_{t=1}^{\infty}\E[\epsilon_t^2] \leq \infty.$$
\end{itemize}
\end{definition}

\begin{lemma} If $\{X_t\}_{t \in \mathbb{N}}$ is a stochastic contraction.
\label{lemma:seriesconvergence}
\begin{itemize}
\item Almost sure convergence $$X_t \xlongrightarrow{a.s.} 0.$$
\item Kolmogorov maximal inequality. For any $T, \delta > 0$ s.t. $\beta = |X_T| + \delta - \sup_{t > T} |\epsilon_t| > 0$,
$$
P\left( \sup_{t \geq T}|X_t| \leq |X_T| + \delta  \right) \geq 1- \frac{4\sum_{t=T+1}^{\infty}\E[\epsilon_t^2]}{\min\{\delta^2, \beta^2\}}.
$$
\end{itemize}
\end{lemma}

\begin{proof} Define the stopping time of sign changes
$$
T_0 = 0, T_k = \inf \{t > T_{k-1}|X_{t-1} \leq 0, X_t > 0 \mbox{ or } X_{t-1} \geq 0, X_t < 0  \}.
$$
We now look at every realized path and examine the segment of the process holding the same sign. W.o.l.g., suppose $X_t \geq 0$ for $T_k < t < T_{k+1}$. Easy to check
\begin{align}
\label{fml:semimrt}
X_t = \E[X_t | \mathcal{F}_{t-1}] + \epsilon_t \leq \lambda_t X_{t-1} + \epsilon_t  \leq X_{t-1} + \epsilon_t \leq X_{T_k} + \sum_{s=T_k+1}^t \epsilon_s.
\end{align}
Therefore $
\left|X_t\right| \leq \left|X_{T_k}\right| + \left|\sum_{s=T_k+1}^t \epsilon_s\right|,
$
same for the negative case. Since $\epsilon_t$'s are independent and $\sum_{t=1}^{\infty}\E[\epsilon_t^2] \leq \infty$, $\sum_{t=1}^{\infty}\epsilon_t$ exists a.s.. Write $N = \sup_k \{T_k \leq \infty\}$ the number of sign changes.

If there are infinitely many sign changes, i.e. $N = \infty$, by sending $k \to \infty$, $\left|X_{T_k}\right| \xlongrightarrow{a.s.} 0$ and $\left|\sum_{s=T_k+1}^{T_k+n} \epsilon_s\right| \xlongrightarrow{a.s.} 0, \forall n > 0. $ Hence $X_t \xlongrightarrow{a.s.} 0$.

If there are finitely many sign changes, we assume w.l.o.g. that for some k, $X_t \geq 0, \forall t \geq T_k$. (\ref{fml:semimrt}) can be written as $X_t - \epsilon_t \leq X_{t-1}$ which indicates $X_t - \sum_{s=T_k+1}^t \epsilon_s$ is decreasing, therefore has a limit ($-\infty$). Since $\sum_{s=T_k+1}^{\infty}\epsilon_s$ exists a.s., $X_t \xlongrightarrow{a.s.} c \geq 0$. Assume $c>0$,
$$
\sum_{s=T_k+1}^{\infty} \epsilon_s \geq \sum_{s=T_k+1}^{\infty} X_s - \lambda_s X_{s-1} = -\lambda_{T_k+1}X_{T_k} + \sum_{s=T_k+2}^{\infty}(1-\lambda_s)X_{s-1} = \infty,
$$
which is a contradiction. Therefore $X_t \xlongrightarrow{a.s.} 0$.

\

To show the maximum inequality, we take the same notations above, and also look at segmentations by sign changes. For any $t$ in the same segment as $T$, $$\left|X_t \right| \leq \left|X_T\right| + \left|\sum_{s=T+1}^t \epsilon_t \right| \leq  \left|X_T\right| + \sup_{T'>T} \left|\sum_{s=T+1}^{T'} \epsilon_s\right|.$$ For any $t$ in a different segment starting at $T'$,
$$
\left|X_t \right| \leq \left| X_T' \right| + \left| \sum_{s=T'+1}^t \epsilon_t \right| \leq  \left|X_T'\right| + \sup_{S>T'} \left|\sum_{s=T'+1}^{S} \epsilon_s\right| \leq \sup_{s > T}|\epsilon_s| + \left|\sum_{s=T'+1}^{S} \epsilon_s\right|.
$$
Now we consider any possible sequence of $\{\epsilon_t, t>T\}$ and allow $T',S$ to change. Kolmogorov maximal inequality implies
$$
P\left(\sup_{i, j>T} \left| \sum_{s=i}^{j} \epsilon_s  \right| \leq x \right) \geq P\left(\sup_{i >T} \left| \sum_{s=T}^{i} \epsilon_s  \right| \leq \frac{x}{2} \right) \geq 1- \frac{4\sum_{s=T}^{\infty}\E[\epsilon_s^2]}{x^2}.
$$
The conclusion is obtained by noticing that $|X_t| \leq |X_T| + \delta$ for any $\{\epsilon_t\}_{t>T}$ satisfying $$\sup_{i, j>T} \left| \sum_{s=i}^{j} \epsilon_s  \right| \leq \min\{\delta, \beta\}.$$
\end{proof}

%%%% SEPARATE %%%%
\subsection{Proof of Theorem \ref{thm:stochasticcontraction}}
\begin{proof}
The idea is to define a sequence of adaptive orthonormal rotations $R_t \in \mathcal{F}_{t-1}$ to align the expected update with the previous step so that we can apply the $\mathbb{R}$ result component-wise. Define $R_t \E[Z_t | \mathcal{F}_{t-1}] = \gamma_{t-1} Z_{t-1},$ for some $\gamma_{t-1}>0, \gamma_{t-1} \in \mathcal{F}_{t-1}$. The contraction assumption also implies that $\gamma_{t-1}\leq \lambda_{t-1}$.
Define a new process $Z^*_i$ satisfying
\begin{enumerate}
\item $Z^*_1 = Z_1, R_1 = I$,
\item writing $R^*_t = \prod_{i=1}^n R_i \in \mathcal{F}_{t-1}$ s.t. $Z^*_t = R^*_t Z_t =  R^*_t \epsilon_t + R^*_t \E[Z_t|\mathcal{F}_{t-1}]. $
\end{enumerate}
Above implies $\norm{Z_t} = \norm{Z^*_t}$, thus we need to prove the equivalence that $Z^*_t \xlongrightarrow{a.s.} 0.$ Notice that
Here $\sum_{i=1}^n R^*_i \epsilon_i$ is component-wise a martingale with
$$
\sum_{i=1}^{\infty} E[\norm{R^*_i \epsilon_i}^2] =  \sum_{i=1}^{\infty} E[\norm{\epsilon_i}^2] < \infty,
$$
hence $\sum_{i=1}^n R^*_i \epsilon_i$ exists a.s.. Since the construction aligns $Z_t^*$ with $\E[Z^*_t|\mathcal{F}_{t-1}]$ we apply Lemma \ref{lemma:seriesconvergence} to obtain almost sure convergence to 0 component-wisely, thus $\norm{Z^*_t} \xlongrightarrow{a.s.} 0$.
\end{proof}

%%%  %%%
\subsection{Proof of Theorem \ref{thm:boulevardconvergence}}
\begin{proof}
Due to non-adaptivity $S_n$ is independent of $\hat{Y}_b$ for any $b$. Notice that $
Y^* = \lambda \E[S_n](Y - Y^*)
$ for $Y^* = \left[\frac{1}{\lambda}I+\E[S_n]\right]^{-1}\E [S_n] Y$.
Define the filtration $\mathcal{F}_b = \sigma(\hat{Y}_0,\dots,\hat{Y}_{b})$ and consider the sequence $Z_b = \hat{Y}_b - Y^*$. This sequence satisfies the stochastic contraction condition.
First, $\norm{Z_0} = \norm{Y^*} \leq \infty$. Notice
\begin{align*}
 \norm{\E[Z_b | \mathcal{F}_{b-1}]} & = \norm{\E\left[ \frac{b-1}{b}  \hat{Y}_{b-1} + \frac{\lambda}{b} S_n(Y - \Gamma_M(\hat{Y}_{b-1})) - Y^*\Big| \mathcal{F}_{b-1} \right]}\\
 & = \norm{\frac{b-1}{b} (\hat{Y}_{b-1}-Y^*) + \frac{\lambda}{b} \E[S_n](Y - \Gamma_M(\hat{Y}_{b-1})) - \frac{1}{b} Y^*}\\
 & \leq \frac{b-1}{b} \norm{\hat{Y}_{b-1}-Y^*} + \norm{\frac{\lambda}{b} \E[S_n](Y - \Gamma_M(\hat{Y}_{b-1})) - \frac{\lambda}{b}\E[S_n](Y - Y^*)} \\
 & \leq \frac{b-1+\lambda}{b} \norm{\hat{Y}_{b-1}-Y^*} \triangleq k_b \norm{Z_{b-1}},
\end{align*}
where $\sum_{b=1}^{\infty} (1-k_b) = \infty$. Since entries and row sums of are both $\leq 1$,
$$
\norm{S_n} \leq \sqrt{\norm{S_n}_{\infty} \norm{S_n}_1} \leq \sqrt{1 \times n} = \sqrt{n}.
$$
Therefore
\begin{align*}
\norm{\epsilon_b} & = \norm{Z_b - \E[Z_b | \mathcal{F}_{b-1}]}
 = \norm{\frac{\lambda}{b} (\E[S_n]-S_n)(Y - \Gamma_M(\hat{Y}_{b-1}))}
 \leq \frac{\lambda}{b} (1 + \sqrt{n}) 2\sqrt{n} M.
\end{align*}
Hence
$$
\sum_{b=1}^{\infty} \E[\norm{\epsilon_b}^2] \leq \left(\sum_{b=1}^{\infty} \frac{1}{b^2}\right) \cdot \lambda^2(1 + \sqrt{n})^2 4nM < \infty.
$$
We conclude that $Z_b \xlongrightarrow{a.s.} 0$, i.e. $\hat{Y}_b \xlongrightarrow{a.s.} Y^*$.
\end{proof}

%%%% SEPARATE %%%%
\subsection{Proof of Corollary \ref{cor:prediction}}
\begin{proof} Expanding $\hat{f}(x)$ gives
\begin{align*}
\hat{f}(x) & = \lim_{B \to \infty} \frac{1}{B}\sum_{b=1}^B s_b(x)(Y-\hat{Y}_b) \\
& = \lim_{B \to \infty} \frac{1}{B}\sum_{b=1}^B s_b(x)(Y-Y^* + Y^* - \hat{Y}_b) \\
& = \lim_{B \to \infty} \frac{1}{B}\sum_{b=1}^B s_b(x)(Y-Y^*) +  \lim_{B \to \infty} \frac{1}{B}\sum_{b=1}^B s_b(x)(Y^* - \hat{Y}_b)\\
& = \E[s_b(x)] (Y - Y^*) + 0 \\
& = \E[s_n(x)] \left[\frac{1}{\lambda}I+\E[S_n]\right]^{-1} Y.
\end{align*}
\end{proof}

%%%% SEPARATE %%%%
\subsection{Proof of Lemma \ref{lemma:expdecay}}
\begin{proof}
\blued{To prove (i), given $k_n$, we focus on the matrix multiplication.} Consider the expansion
\begin{align*}
\left[\frac{1}{\lambda}I + K_n\right]^{-1} = \lambda \sum_{i=0}^{\infty}\left((\lambda)^{2i}K_n^{2i} - (\lambda)^{2i+1}K_n^{2i+1}\right).
\end{align*}
We examine the column sums of each of the matrix powers. Start with $K_n^2$,
$$
\sum_{i=1}(K^2_n)_{i,1} = \sum_{i=1}^n \sum_{j=1}^n (K_n)_{i,j}(K_n)_{j,1} = \sum_{j=1}^n (K_n)_{j,1}\sum_{i=1}^n(K_n)_{i,j}.
$$
Since $K_n$ consists of structure vectors of sample points, for some $c>0$,
$$
1-\frac{c}{n} \leq \sum_{j=1}^n (K_n)_{i,j} = \sum_{j=1}^n (K_n)_{i,j} \leq 1, \quad i = 1, \dots, n.
$$
Given $K_n$ is nonnegative,
$$
\left(1-\frac{c}{n}\right)^2 \leq \sum_{i=1}(K^2_n)_{i,1} = \sum_{j=1}^n (K_n)_{j,1}\sum_{i=1}^n(K_n)_{i,j} \leq 1.
$$
Repeating the same discussion yields
$$
\left(1-\frac{c}{n}\right)^m \leq \sum_{i=1}(K^m_n)_{i,1}\leq 1.
$$
Therefore,
\begin{align*}
\lambda \left(\frac{1}{1-\lambda^2(1-\frac{c}{n})^2} - \frac{\lambda}{1-\lambda^2} \right) & \leq \sum_{j=1}^n \left[\frac{1}{\lambda}I + K_n\right]^{-1}_{j,1} \\
& = \lambda \left(\sum_{i=0}^{\infty}(\lambda)^{2i}(K_n^{2i})_{j,1} - (\lambda)^{2i+1}(K_n^{2i+1})_{j,1} \right) \\
& \leq \lambda \left(\frac{1}{1-\lambda^2} - \frac{\lambda}{1-\lambda^2(1-\frac{c}{n})^2} \right),
\end{align*}
where both the LHS and RHS reduce to $\frac{\lambda}{1+\lambda} + \O{\frac{1}{n}}$. So is true for any column sum of $\left[\frac{1}{\lambda}I + K_n\right]^{-1}$. Now given $k_n$ is nonnegative and $1 - \norm{k_n}_1 = \O{\frac{1}{n}}$ we reach the assertion.

\

%%%% SEPARATE %%%%
\blued{To show (ii), under} locality, $k_{nj} = 0$ if $|x_i- x_j| >  d_n$, while $[K_n]_{i,j}=0$ if $|x_i- x_j| >  d_n$. Recursively, if $|x_i- x_j| > l_n \cdot  d_n$ then $[K_n^l]_{i,j} = 0$ for $ l \leq l_n$. As $k_n$ and $K_n$ are element-wisely nonnegative, we again expand the matrix inverse
\begin{align*}
\norm{r_n\proj{D_n^c}}_1 =  \sum_{|x - x_i|>l_n \cdot  d_n}|r_{ni}|
= &  \sum_{|x - x_i|>l_n \cdot  d_n} \left|\sum_j k_{nj} \left[\frac{1}{\lambda}I + K_n\right]^{-1}_{j,i} \right| \\
= & \sum_{|x - x_i|>l_n \cdot  d_n} \left|\sum_{|x - x_j| \leq d_n} k_{nj} \left[\frac{1}{\lambda}I + K_n\right]^{-1}_{j,i} \right| \\
\leq & \sum_{|x - x_j| \leq  d_n} k_{nj} \sum_{|x - x_i|>l_n \cdot  d_n} \left| \left[\frac{1}{\lambda}I + K_n\right]^{-1}_{j,i} \right| \\
\leq &  \sum_{|x - x_j| \leq d_n} k_{nj} \sum_{|x_i - x_j| > (l_n-1)\cdot d_n}\left| \left[\frac{1}{\lambda}I + K_n\right]^{-1}_{j,i} \right| \\
\leq & \sum_{|x - x_j| \leq  d_n} k_{nj} \sum_{|x_i - x_j| > (l_n-1)\cdot d_n} \lambda \sum_{l=l_n}^{\infty} \lambda^l [{K_n^l}]_{j,i} \\
\leq & \sum_{|x - x_j| \leq  d_n} k_{nj} \sum_{l=l_n}^{\infty} \lambda^{l+1} \\
\leq & \sum_{l=l_n}^{\infty} \lambda^{l+1} =  \frac{\lambda}{1-\lambda}\frac{1}{n}.
\end{align*}
\end{proof}

%%%%
\subsection{Proof of Lemma \ref{lemma:rateofk}}
\begin{proof}
The idea is to bound $k_{nj}$ from both above and below. The condition $$ \inf_{A \in q \in Q_n} \sum_{i = 1}^n I(x_i \in A)  = \Ome{n^{\frac{1}{d+2}} } $$
implies that $k_{nj} = \O{n^{-\frac{1}{d+2}}}$. Given $\norm{k_n}_1 \leq 1$,
$$
\norm{k_n} \leq \sqrt{\norm{k_n}_1 \norm{k_n}_{\infty}} = \O{n^{-\frac{1}{2}\frac{1}{d+2}}}.
$$
On the other hand, given $\left| B_n \right|  = \Ome{ n \cdot d_n^d }$, there are at most
$$
\Ome{n \cdot d_n^d} = \Ome{n^{\frac{1}{d+1}}}
$$
$k_{nj}$'s that are positive. Since $\norm{k_n}_1 = 1-O(n^{-1})$, $$
\norm{k_{n}} = \Ome{ \sqrt{\left(n^{-\frac{1}{d+1}} \right)^2 \cdot n^{\frac{1}{d+1}}}} =  \Ome{n^{-\frac{1}{2}\frac{1}{d+1}}}.
$$
Those bounds also work for $\norm{r_n}$ given
$$
\frac{\lambda}{1+\lambda} \leq eigen\left(\left[\frac{1}{\lambda}I + K_n\right]^{-1}\right) \leq \lambda.
$$
\end{proof}

%%%  %%%
\subsection{Proof of Theorem \ref{thm:fixed}}
We first introduce a lemma regarding $|B_n|$ and $|D_n|$. They follow Binomial distributions with parameters depending on the local covariate density.
\begin{lemma}
Assume $X_1, X_2,\dots$ independent binomial random variables s.t. $X_n \sim Binom(n, p_n)$ and $np_n \to \infty$.
$$
\frac{X_n - np_n}{\sqrt{np_n(1-p_n)}} \xlongrightarrow{d} N(0,1).
$$
\end{lemma}
\blued{The proof of Theorem \ref{thm:fixed} follows.}
\begin{proof}
Notice that
\begin{align*}
\hat{f}_n(x) - r_n^Tf(X_n) =r_n^T \vec{\epsilon}_n.
\end{align*}
\chgx{To obtain a CLT} we check the Lindeberg-Feller condition of $r_n^T \vec{\epsilon}_n$, i.e. for any $\delta > 0$,
$$
\lim_n \frac{1}{\norm{r_n}^2 \sigma_{\epsilon}^2}\sum_{i=1}^n  \E \left[(r_{ni} \epsilon_i)^2 I(|r_{ni}\epsilon_i| > \delta \norm{r_n}\sigma_{\epsilon}) \right] = 0.
$$
Since $\norm{k_n}_{\infty} = O\left(n^{-\frac{1}{d+2}}\right)$ and $\left[\frac{1}{\lambda}I + K_n\right]^{-1}$ having row sums of $\frac{\lambda}{1+\lambda}+\O{n^{-1}}$, we have
$$
\norm{r_n}_{\infty} \leq \norm{k_n}_{\infty} \cdot \norm{\left[\frac{1}{\lambda}I + K_n\right]^{-1}}_1 = O\left(n^{-\frac{1}{d+2}}\right).
$$
Furthermore, since $\norm{r_n} = \Ome{n^{-\frac{1}{2}\frac{1}{d+1}}}$, we get
$$
\frac{\norm{r_n}_{\infty}}{\norm{r_n}} = O\left( n^{-\frac{1}{d+2} + \frac{1}{2}\frac{1}{d+1}}\right),
$$
which justifies the Lindeberg-Feller condition when $\epsilon$ is sub-Gaussian by
\begin{align*}
\sum_{i=1}^n  \E \left[(r_{ni} \epsilon_i)^2 I(|r_{ni}\epsilon_i| > \delta \norm{r_n}\sigma_{\epsilon}) \right]
& \leq \sum_{i=1}^n r_{ni}^2 \sqrt{\E[\epsilon_i^4]\cdot \E[I(|r_{ni}\epsilon_i| > \delta \norm{r_n}\sigma_{\epsilon})^2]}  \\
& \leq \sum_{i=1}^n r_{ni}^2 \sqrt{\E[\epsilon_i^4]} \cdot \sqrt{P\left( |\epsilon_i| \geq \frac{\delta \| r_n \| \sigma_{\epsilon}}{r_{ni}} \right)} \\
& \leq \sum_{i=1}^n r_{ni}^2 \sqrt{\E[\epsilon_i^4]} \sqrt{2 \exp\left(-\frac{1}{2\sigma_{\epsilon}^2}\cdot \left(\frac{\delta \| r_n \| \sigma_{\epsilon}}{r_{ni}} \right)^2\right)} \\
& \leq \norm{r_n}^2 \exp \left( - \O{n^{\frac{2}{d+2}-\frac{1}{d+1}}\right)} \longrightarrow 0,
\end{align*}
since
$$
P(\epsilon > t) \leq \exp\left(-\frac{t^2}{2\sigma^2_{\epsilon}}\right)
$$
for sub-Gaussian $\epsilon$.
\end{proof}

%%%  %%%
\subsection{Kolmogorov's Extension Theorem}
\label{sec:ket}
Define $(x_1,\dots) = \X \in [0,1]^{d\times \mathbb{N}}$ and $\epsilon = (\epsilon_1,\dots) \in \mathbb{R}^\mathbb{N}$, where the probability measures on $[0,1]^{d\times \mathbb{N}}$ and $\mathbb{R}^\mathbb{N}$ are uniquely decided by the product measures on the cylinder spaces reflecting i.i.d. sampling i.e. $y_i = f(x_i) + \epsilon_{i}$ for $i \in \mathbb{N}$. Write $\pi_i$ the cumulative coordinate projection, i.e. $\pi_i(a_1,\dots,a_n,\dots) = (a_1,\dots,a_i). $ We can calculate $k_n$ and $K_n$ with respect to $\Pi_n = (\pi_n(\X), \pi_n(\epsilon))$. Thus
$$
\rho_n(\X, \epsilon) = \frac{\hat{f}_n(x; \Pi_n) - k_n^T(x; \Pi_n)[\frac{1}{\lambda}I + K_n(\Pi_n)]^{-1}f(\Pi_n)}{\norm{k_n(x;\Pi_n)^T[\frac{1}{\lambda}I + K_n(\Pi_n)]^{-1}}}
$$
reflects the prediction after using a random sample of size $n$. Using Lemma \ref{lemma:fixtorandom}, CLT of $\rho_n$ requests an almost surely claim of Theorem \ref{thm:fixed} where the sequence of $(x_1, y_1), \dots,(x_n, y_n)$ comes from $(\pi_n (\X), \pi_n (\epsilon))$.

We also benefit from the following lemma.
\begin{lemma}
\label{lemma:fixtorandom}
Assume $X: \Omega_1 \to S$, independent of $\epsilon: \Omega_2 \to S,\{f_n: S\times S \to \mathbb{R} \}$ sequence of measurable functions. Assuming for a.s. $x \in \Omega_1$,
$$
f_n(x, \epsilon) \xrightarrow{d} N(0,1).
$$
Then
$$
f_n(X, \epsilon) \xrightarrow{d} N(0,1).
$$
\end{lemma}
\begin{proof}
Probabilistic DCT guarantees that
\begin{align*}
\lim_n P(f_n(X, \epsilon) \leq t) & =\lim_n \int \int\mathbbm{1}_{\{f_n(x, \epsilon) \leq t\}} d\mu_x d\mu_{\epsilon}  \\
&= \lim_n \int P(f_n(x, \epsilon) \leq t) d\mu_x \\
&= \int \lim_n P(f_n(x, \epsilon) \leq t) d\mu_x \\
& = \int \Phi(t) d\mu_x = \Phi(t).
\end{align*}
\end{proof}

%%%%

\subsection{Proof of Lemma \ref{lemma:random}}
\label{sec:lemma:random}
\begin{proof}
In order to prove the lemma, we combine Lemma \ref{lemma:rateofk}, Theorem \ref{thm:fixed} and Lemma \ref{lemma:fixtorandom} and show that all assumptions are met from a point-wise perspective on $[0,1]^{d\times\mathbb{N}}$, i.e. fixed sample sequence are given by $\pi_n \X, n\geq 1$.

\

i) We show for a.s. $\X$, $\left| B_n \right| = \left| B_n(\pi_n \X) \right|= \Ome{ n \cdot d_n^d }$. Consider random $\X$. \blued{Since $\mu(x)$ is bounded away from $0$, \blued{the expected amount of sample points in $B_n$ should be at least proportional to its volume $d_n^d$. } Define $a_n$ s.t. $na_n=\E[\left| B_n \right|] = \Ome{nd_n^d}$.} Considering \blued{the multiplicative Chernoff Bound applied to binomial random variables. For fixed $0< c < 1$,

\begin{align*}
P\left( |B_n| \leq c \cdot n a_n\right) \leq & \exp \left( -\frac{(1-c)^2 na_n}{2}\right).
\end{align*}
Further, since $na_n = \Ome{nd_n^d} = \Ome{n^{\frac{1}{d+1}}}$, there exists some $M>0$ s.t.
\begin{align*}
\sum_{n=1}^{\infty}  \exp\left(-\frac{(1-c)^2na_n}{2}\right) &\leq \int_{n=0}^{\infty} \exp\left(-\frac{(1-c)^2M}{2}\cdot n^{\frac{1}{d+1}}\right) dn\\
& \leq \int_{u=0}^{\infty}  \exp\left(-\frac{(1-c)^2M}{2}\cdot u \right)(d+1)u^d du\\
& < \infty.
\end{align*}
}
As per Borel-Contelli, since
$$
\sum_{n=1}^{\infty} P(|B_n(\pi_n \X)| \leq c \cdot na_n ) < \infty,
$$
then for a.s. $\X$, events of $|B_n(\pi_n \X)| \leq  c \cdot na_n$ happens finite times, \blued{which leads to our conclusion.}

\

ii) To show $$\inf_{A \in q \in Q_n} \sum_{i = 1}^n I(x_i \in A) = \Ome{n^{\frac{1}{d+2}}}$$ for a.s. $\X$, apply the concentration bound again to get:\blued{
\begin{align*}
P\left(\exists A \in q \in Q_n \mbox { s.t.} \sum_{i=1}^n I(x_i \in A) \leq n^{\frac{1}{d+2}}\right)
& \leq |Q_n| \cdot n^{\frac{d+1}{d+2}}\cdot P\left(  \sum_{i=1}^n I(x_i \in A) \leq n^{\frac{1}{d+2}} \right) \\
& \leq |Q_n| \cdot n^{\frac{d+1}{d+2}}\cdot \exp\left(-\frac{n^{\frac{1}{d+2} + \nu} \cdot n^{-2\nu}}{2}\right) \\
& \leq \frac{1}{n}\exp\left( \frac{1}{2} n^{\frac{1}{d+2}-\nu} - n^{\alpha}\right) \cdot n \cdot \exp\left(-\frac{1}{2}n^{\frac{1}{d+2}-\nu}\right)\\
& \leq \exp\left(-n^{\alpha}\right).
\end{align*}
Therefore, noticing that for $\alpha>0$,
$$
\sum_{n=1}^{\infty} \exp\left(-n^{\alpha}\right)< \infty,
$$}
the Borel-Cantelli theorem indicates our assertion. Hence, for a.s. $\X$, \blued{the triangular array} $\pi_n \X$ satisfies the assumptions in Theorem \ref{thm:fixed}.
\end{proof}

%%%%%%
\subsection{Proof of Theorem \ref{thm:main} (Main Theorem)}
\begin{proof}
We first show that for given $x \in [0,1]^d$,
$$
\frac{r_n^Tf(X_n) - \frac{\lambda}{1+\lambda}f(x)}{\norm{r_n}} \xlongrightarrow{p} 0.
$$
Recall the index set $D_n = \{i : |x_i-x| \leq l_n\cdot d_n \}$. Denote $\Delta = \frac{\lambda}{1+\lambda} - \sum_{i=1}^n r_{n,i} = \O{n^{-1}}$ and $\tilde{f}(x) = (f(x),\dots, f(x))^T$ an $n$-vector. We split
\begin{align*}
\frac{r_n^Tf(X_n) - \frac{\lambda}{1+\lambda}f(x)}{\norm{r_n}} & =  \frac{r_n^T[f(X_n) - \tilde{f}(x)]}{\norm{r_n}} - \frac{\Delta \cdot f(x)}{\norm{r_n}} \\
= & - \frac{\Delta \cdot f(x)}{\norm{r_n}} + \frac{\left(r_n\proj{D_n}\right)^T [f(X_n) - \tilde{f}(x)]\proj{D_n}}{\norm{r_n}} + \frac{\left(r_n\proj{D_n^c}\right)^T [f(X_n) - \tilde{f}(x)]\proj{D_n^c}}{\norm{r_n}}.
\end{align*}
By replacing $\O{\cdot}$ in the fixed case by $\Op{\cdot}$ in the random design case, recall that
$$
\Omep{ n^{-\frac{1}{2}\frac{1}{d+1}}} = \norm{k_n}, \norm{r_n} = \Op{n^{-\frac{1}{2}\frac{1}{d+2}}}.
$$
On one hand, we notice that
$$
\left| \left(r_n\proj{D_n^c}\right)^T [f(X_n) - \tilde{f}(x)]\proj{D_n^c}\right| \leq \norm{r_n\proj{D_n^c}}_1 \cdot \norm{[f(X_n) - \tilde{f}(x)]\proj{D_n^c}}_{\infty} = \Op{\frac{1}{n} \cdot 2M_f} = \Op{n^{-1}}.
$$
Therefore
$$
\frac{\left(r_n\proj{D_n^c} \right)^T  [f(X_n) - \tilde{f}(x)]\proj{D_n^c}}{\norm{r_n}} \xlongrightarrow{p} 0.
$$
And similarly since $| \Delta | = \O{n^{-1}}$,
$$
\frac{\Delta \cdot f(x)}{\norm{r_n}} \xlongrightarrow{p} 0.
$$
On the other hand, we can show similarly as $|B_n|$ that $|D_n| = \Omep{n \cdot (l_n\cdot d_n)^d)}$ a.s. and therefore
\begin{align*}
\frac{\left| \left(r_n\proj{D_n} \right)^T  [f(X_n) - \tilde{f}(x)]\proj{D_n}\right|}{\norm{r_n}} & \leq \frac{\norm{r_n\proj{D_n}}\cdot\norm{[f(X_n) - \tilde{f}(x)]\proj{D_n}}}{\norm{r_n}} \\
& \leq \norm{[f(X_n) - \tilde{f}(x)]\proj{D_n}} \\
& = \Op{\sqrt{n\cdot (l_nd_n)^d\cdot (l_nd_n\cdot{\alpha})^2}} \\
& = \Op{\sqrt{n\cdot \log_n^{d+2} \cdot d_n^{d+2}}} \\
& = \Op{\sqrt{n\cdot \log_n^{d+2} \cdot n^{-\frac{d+2}{d+1}}}}\\
& = \Op{(\log n)^{\frac{d+2}{2}} n^{-\frac{1}{2}\frac{1}{d+1}}}.
\end{align*}
Therefore
$$
\frac{\left(r_n\proj{D_n}\right)^T [f(X_n) - \tilde{f}(x)]\proj{D_n}}{\norm{r_n}} \xlongrightarrow{p} 0.
$$
Combining \chgx{the above calculations} gives the result that
$$
\frac{r^T_nf(X_n) - \frac{\lambda}{1+\lambda}f(x)}{\norm{r_n}} \xlongrightarrow{p} 0.
$$
Therefore by Slutsky's Theorem,
$$
\frac{\hat{f}_n(x) - \frac{\lambda}{1+\lambda}f(x)}{\norm{r_n}} = \frac{\hat{f}_n(x) - r_n^Tf(X_n)}{\norm{r_n}} + \frac{r_n^Tf(X_n) - \frac{\lambda}{1+\lambda}f(x)}{\norm{r_n}}\xlongrightarrow{d} N(0,\sigma^2_{\epsilon}).
$$
\end{proof}

%%%%%%
\subsection{Proof of Theorem \ref{thm:escape}}
\begin{proof} We refer to Theorem \ref{thm:stochasticcontraction}. Choose $\delta = r$, and choose $T$ s.t. $\forall t > T,$ $$\frac{\lambda}{t} (1 + \sqrt{n}) 2\sqrt{n} M \leq \frac{r}{\sqrt{d}}, \mbox{ i.e. } \sup \norm{\epsilon_t} \leq \frac{r}{\sqrt{d}},$$
In this case, $\beta = \norm{\hat{Y}_t} + \delta - \sqrt{d}\sup_{t\geq T}\norm{\epsilon_t} \geq \delta = r$. By the conditional independence of $\hat{Y}_t$ and $\epsilon_b, b>t$ in the contraction region,
\begin{align*}
P\left( \hat{Y}_b \in C, \forall b \geq t \big| \hat{Y}_t \in B(Y^*, r)\right) & \geq P\left( \sup_{b > t}\norm{\hat{Y}_b - Y^*} \leq \norm{\hat{Y}_t - Y^*} + \delta   \Big| \hat{Y}_t \in B(Y^*, r)\right)\\
& = P\left( \sup_{b > t}\norm{\hat{Y}_b - Y^*} \leq \norm{\hat{Y}_t - Y^*} + \delta \right)\\
& \geq 1- \frac{4\sqrt{d}\sum_{b=t+1}^{\infty}\E[\epsilon_b^2]}{r^2} \longrightarrow 1.
\end{align*}
\end{proof}

\section{On Empirical Study}
\label{sec:emp}
\subsection{Empirical Study Details}
\label{sec:appemp}
\begin{itemize}
\item Figure \ref{fig:dist_truth} shows the distribution of true responses for our simulated training set.
\item \blued{The 4 test points used for comparing Boulevard and kernel ridge regression are:
(0.1, 0.1, 0.1, 0.1, 0.1), (0.6, 0.9, 0.8, 0.9, 0.7), (0.1, 0.1, 0.9, 0.9, 0.9), (0.9, 0.1, 0.1, 0.1, 0.9).}
\item The 10 test points \blued{used for showing limiting distributions and generating reproduction intervals} are:
(0.5, 0.5, 0.5, 0.5, 0.5),
(0.2, 0.2, 0.2, 0.2, 0.2),
(0.1, 0.9, 0.1, 0.9, 0.1),
(0.1, 0.1, 0.9, 0.9, 0.9),
(0.9, 0.1, 0.1, 0.1, 0.9),
(0.5, 0.1, 0.9, 0.1, 0.5),
(0.3, 0.2, 0.7, 0.8, 0.6),
(0.4, 0.2, 0.3, 0.6, 0.7),
(0.2, 0.7, 0.8, 0.3, 0.5),
(0.3, 0.6, 0.4, 0.9, 0.5).
\item Table \ref{tab:es} shows the settings we use in studies. The labels are: MSE for Figure \ref{fig:sim_syn}, MSE-DatasetName for Figure \ref{fig:sim_real}, Limiting for Figure \ref{fig:pred_dist}, Variance for Table \ref{tab:pd_var} and RI for Figure \ref{fig:pi}. Abbreviations: $n$ for sample size, $\theta$ for subsample rate, ntree for ensemble size, $k$ for terminal leave size after subsample (which has to be corrected when no subsample is involved, i.e. GBT), $\lambda$ as in Boulevard iterations.
\end{itemize}

\begin{figure}[htbp] %  figure placement: here, top, bottom, or page
   \centering
   \includegraphics[width=0.6\textwidth]{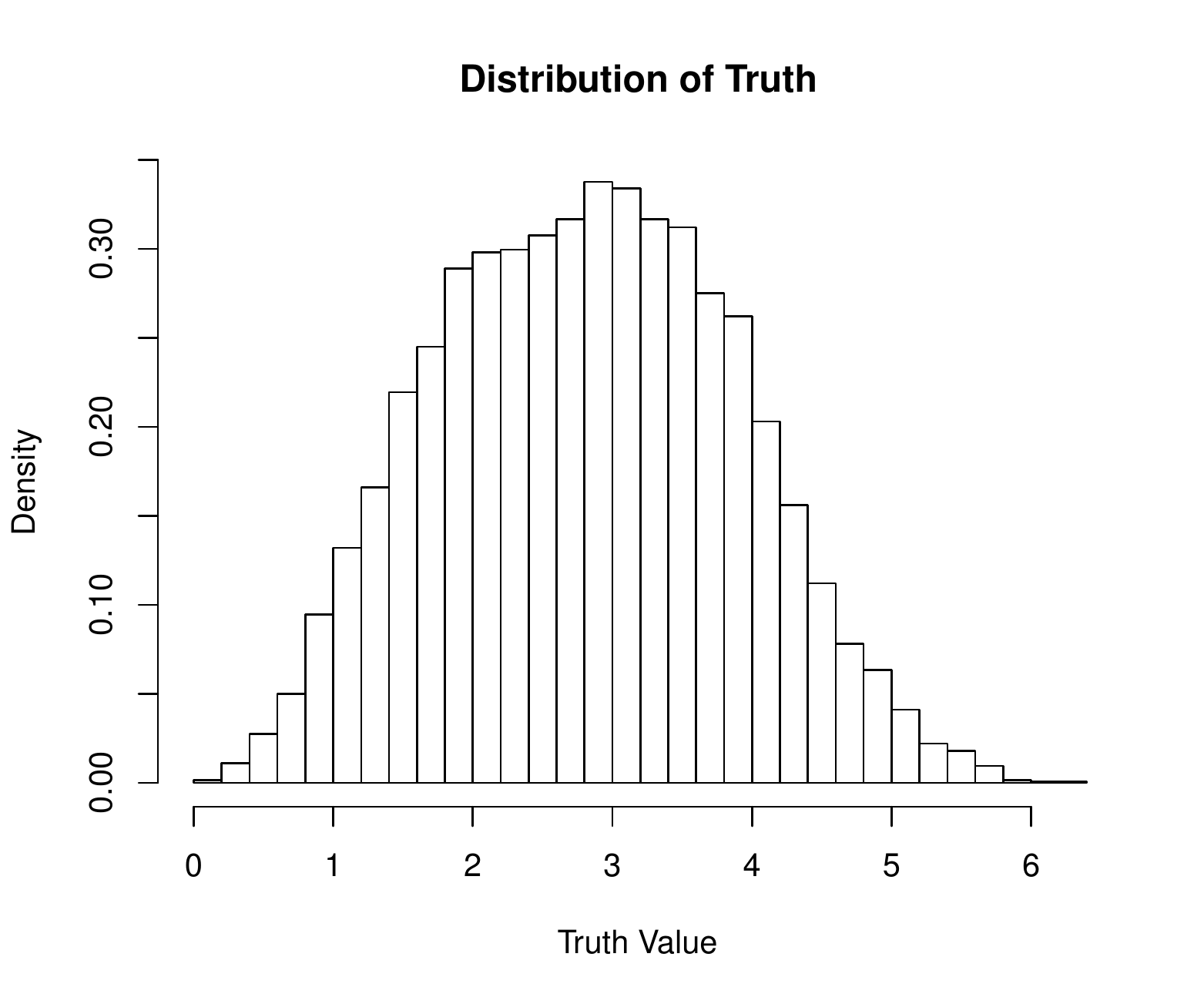}
   \caption{Distribution of truth used when assessing limiting distribution and reproduction interval.}
   \label{fig:dist_truth}
\end{figure}

% Requires the booktabs if the memoir class is not being used
\begin{table}[htbp]
   \centering
   %\topcaption{Table captions are better up top} % requires the topcapt package

   \begin{tabular}{cccccc} % Column formatting, @{} suppresses leading/trailing space
      \toprule
	 label 				& n 		& $\theta$ & ntree 	& 		k 		& $\lambda$  \\
	 \midrule
	 MSE-(1-4) 		& 5000	& 0.3 		& 1000		& 		20		& 0.8	 			\\
	 MSE-Boston 	& 506	& 0.8 & 1000 & 5 & 0.8 \\
	 MSE-CCPP 		& 9568 	& 0.5	 &  1000 & 50 &  0.8 \\
	 MSE-CASP 		&20000	& 0.5			& 1000		& 	50			& 0.8			\\
	 MSE-Airfoil 		& 1503	&	0.8		& 1000 		&  40				& 0.8			\\
	 Limiting-(1-4) 	& 1000 & 0.8 & 2000 & 10 & 0.5 \\
	 Variance-(1-4) 	& 5000 & 0.8 & 3000 & 20 & 0.5 \\
	 RI-(1-2)			& 1000 & 0.8 & 2000 & 10 & 0.5 \\
	 RI-(3-4)			& 5000 & 0.8 & 2000 & 10 & 0.5 \\
     \bottomrule
   \end{tabular}
   \caption{Parameters used in empirical study.}
   \label{tab:es}
\end{table}

\subsection{Scaling of Prediction Error with Response Error}
Table \ref{tab:pd_var} shows the experiment in which we apply symmetric uniform errors and observe the scaling of prediction standard deviation along with the increase of error standard deviation.
\begin{table}[htbp]
   \centering
   %\topcaption{Table captions are better up top} % requires the topcapt package
   \begin{tabular}{c|cccccccccc} % Column formatting, @{} suppresses leading/trailing space
      \toprule
      Error\textbackslash  Fixed Point& 1&2&3&4&5&6&7&8&9&10\\
      \midrule
0 &0.030&0.044&0.044&0.049&0.050&0.037&0.038&0.033&0.032&0.040 \\
Unif[-1,1] &0.067&0.089&0.096&0.087&0.096&0.083&0.081&0.074&0.071&0.082 \\
Unif[-2,2] &0.119&0.154&0.172&0.158&0.162&0.152&0.122&0.139&0.137&0.145 \\
Unif[-4,4] &0.243&0.271&0.278&0.278&0.288&0.317&0.284&0.289&0.318&0.254 \\
      \bottomrule
   \end{tabular}
   \caption{Prediction standard deviations scale with error standard deviations.}
   \label{tab:pd_var}
\end{table}

\end{document}